\documentclass[conference]{IEEEtran}

\usepackage{amsmath,amssymb,amsthm}
\usepackage{algorithm}
\usepackage{algpseudocode}
\usepackage{graphicx}
\usepackage{cite}
\usepackage{multirow}
\usepackage{booktabs}
\usepackage{xcolor}

\newtheorem{theorem}{Theorem}
\newtheorem{corollary}{Corollary}
\newtheorem{assumption}{Assumption}

\begin{document}
	
	\title{Large-Scale Partition-Based RIS Beamforming For Uplink RIS-Equipped Multi-User Systems: Asymptotic Analysis}
	
	\author{
		\IEEEauthorblockN{Samira Rahimian, Haris Gacanin}
		\IEEEauthorblockA{\textit{Chair for Distributed Signal Processing}\\
			\textit{RWTH Aachen University, Aachen, Germany}\\
			Email:\{rahimian, harisg@dsp.rwth-aachen.de\}}}
	
	\maketitle
	
	\begin{abstract}
		Combining a reconfigurable intelligent surface (RIS) with a receive antenna array is a promising low-complexity architecture for multi-user uplink reception, but its performance analysis for more than two users has remained an open problem: the zero-forcing (ZF) signal-to-interference-plus-noise ratio (SINR) no longer admits an explicit, low-dimensional closed-form, and its distribution is analytically intractable for design purposes. This paper closes that gap for a $K$-user, $K$-antenna uplink system in which a large-scale, $L$-element RIS, partitioned into $K$ user-dedicated sub-surfaces, precedes ZF reception at the base station. Through an asymptotic, large-scale analysis in which the sub-surface sizes grow without bound, we show that the orthogonal projector underlying the ZF SINR converges to a rank-one matrix aligned with the desired user's channel, and we use this convergence to derive a closed-form asymptotic approximation for the average per-user SINR that depends only on deterministic channel parameters. Treating this expression as a tractable design objective, we prove that equal partitioning is approximately sum-rate-optimal at leading order regardless of path-loss asymmetry across users, and we develop a low-complexity greedy pairwise-transfer search that refines the partition beyond this leading-order optimum. Monte Carlo simulations across a range of system and RIS sizes confirm that the closed-form SINR tracks the exact simulated rate closely once the RIS is large relative to the number of users, and that the greedy search yields consistent, if modest, sum-rate gains over equal partitioning, validating the theory as both an accurate performance predictor and a practical design tool.
	\end{abstract}
	
	
	\section{Introduction}
	The push toward ever-higher data rates and stricter reliability
	targets in wireless networks has driven interest in surfaces
	that can actively reshape the propagation environment rather
	than merely receive or relay through it. The reconfigurable
	intelligent surface (RIS) is a leading instance of this idea:
	an array of low-cost, largely passive elements, each able to
	impose an independent phase (and, in some designs, amplitude)
	shift on an incident electromagnetic wave, collectively
	turning the wireless channel into a parameter that can be
	engineered rather than merely
	endured~\cite{DiRenzo2020,Zhao2025,Huang2020Holographic}.
	
	RIS-assisted architectures have been studied in two broad
	deployment settings. On the transmit side, a base station (BS)
	can pair a compact feed antenna, or a small number of feeds,
	with an RIS to synthesize multiple data streams for concurrent
	users, an approach pursued
	in~\cite{Li2025Transmissive,Li2025TMA,Mao2024Index,Huang2024Hybrid},
	while~\cite{Zeng2024Reconfigurable} extends this to a
	multi-feed transmitter and jointly analyzes the resulting
	system capacity and energy efficiency. Under non-orthogonal
	multiple access (NOMA), a related line of
	work~\cite{SunJing2024Holographic} designs RIS phase profiles
	that minimize the transmit power needed to meet per-user
	quality-of-service targets, for transmitters with either a
	single antenna or an array.
	
	A second, receiver-centric line of work places the RIS next to
	the receiving terminal instead. One branch equips every
	surface element with its own radio-frequency chain, giving a
	fully digitally processed array often termed a large
	intelligent surface
	(LIS)~\cite{Jung2019Reliability,Yuan2020Towards,Jung2020Performance,Jung2021PerformanceActive,Hu2018Beyond,Sanchez2022Distributed}.
	A second, more hardware-frugal branch pairs a passive RIS with
	a small conventional antenna
	array~\cite{Li2024TowardTransmissive,An2024Exploiting,Sun2025SINR}.
	Within this branch, \cite{Li2024TowardTransmissive} jointly
	optimizes power allocation, subcarrier assignment, and RIS
	phase coefficients to maximize the sum-rate of an uplink
	multi-user system subject to individual rate guarantees, with
	the receiver formed by an RIS and a horn antenna;
	\cite{An2024Exploiting} targets high-altitude-platform (HAP)
	networks with simultaneous wireless information and power
	transfer, where the HAP serves multiple ground terminals each
	equipped with a multi-layer refracting RIS and an antenna; and
	\cite{Sun2025SINR} studies a continuous-aperture RIS paired
	with a single antenna, deriving an SINR analysis and an
	accompanying interference-mitigation scheme for multi-user
	reception.
	
	Our own prior work~\cite{RahimianJingYu2026} falls in this
	second, receiver-centric category, but targets the smallest
	non-trivial multi-user instance: two single-antenna users
	served by a receiver formed from an RIS and two antennas. The
	surface there is split into two sub-surfaces, one per user,
	each shaping its dedicated user's incident signal, following
	the partitioning philosophy already explored
	in~\cite{SunJing2024Holographic,Chauhan2023RIS,KhaleelBasar2022NOMA};
	zero-forcing (ZF) at the two receive antennas then clears
	whatever interference the partitioning leaves behind. An
	asymptotic statistical treatment of the resulting SINR yielded
	closed-form per-user rate approximations, which simulation
	confirmed to be accurate and to translate into a tangible
	performance gain over unpartitioned designs.
	
	The present paper lifts that two-user, two-antenna analysis in~\cite{RahimianJingYu2026} to an arbitrary $K$-user, $K$-antenna uplink system served by a large-scale RIS with $L$ elements, partitioned $K$ ways, and the generalization is not merely a change of index range. In the $2\times2$ case, the ZF SINR reduces to an explicit ratio of $2\times2$ determinants that can be handled term by term; for general $K$, no such closed-form exists, and the SINR instead hinges on the orthogonal projector onto the complement of the interfering users' joint channel subspace. Our main technical contribution is an asymptotic, large-scale analysis showing that, as each sub-surface size grows without bound, this projector concentrates around a rank-one matrix, a convergence argument with no counterpart in the two-user analysis. That result then unlocks a closed-form asymptotic approximation, valid in this large-scale regime, for the average per-user ZF SINR, expressed purely in terms of the deterministic parameters of the channels. Treating this expression as a design objective, we further prove that equally sized partitions are approximately sum-rate-optimal at leading order, regardless of how asymmetric the users' path losses are, and we benchmark a low-complexity greedy pairwise-transfer search against equal partitioning under realistic Monte Carlo simulation.
	\section{System Model}
	\label{sec:system}
	We consider a communication system where $K$ single-antenna users transmit information to a receiver equipped with an array of $K$ antennas and an adjacent Reconfigurable Intelligent Surface (RIS). The receiver antennas form a uniform linear array (ULA). The RIS consists of $L$ elements arranged on a uniform planar array (UPA) of $M\times M$ elements ($L=M^2$). The center of the RIS is located at the origin $(0,0,0)$ of a 3D Cartesian
	coordinate system, and the RIS is located in the
	$(y,z)$-plane. The RIS element in row $m$ and column $p$
	($m=0,\ldots,M-1,p=0,\ldots,M-1$) is denoted the $n$-th element under the column-major flattening $n = pM+m+1$.
	With inter-element spacing $d_S$ in both the $y$ and $z$ directions, the
	coordinate of the $n$-th RIS element is $\mathbf{p}_n = \left(0,\;
	\Bigl(m-\tfrac{M-1}{2}\Bigr)d_S,\;
	\Bigl(p-\tfrac{M-1}{2}\Bigr)d_S\right).$
	\begin{figure}[t]
		\centering
		\includegraphics[width=0.65\linewidth]{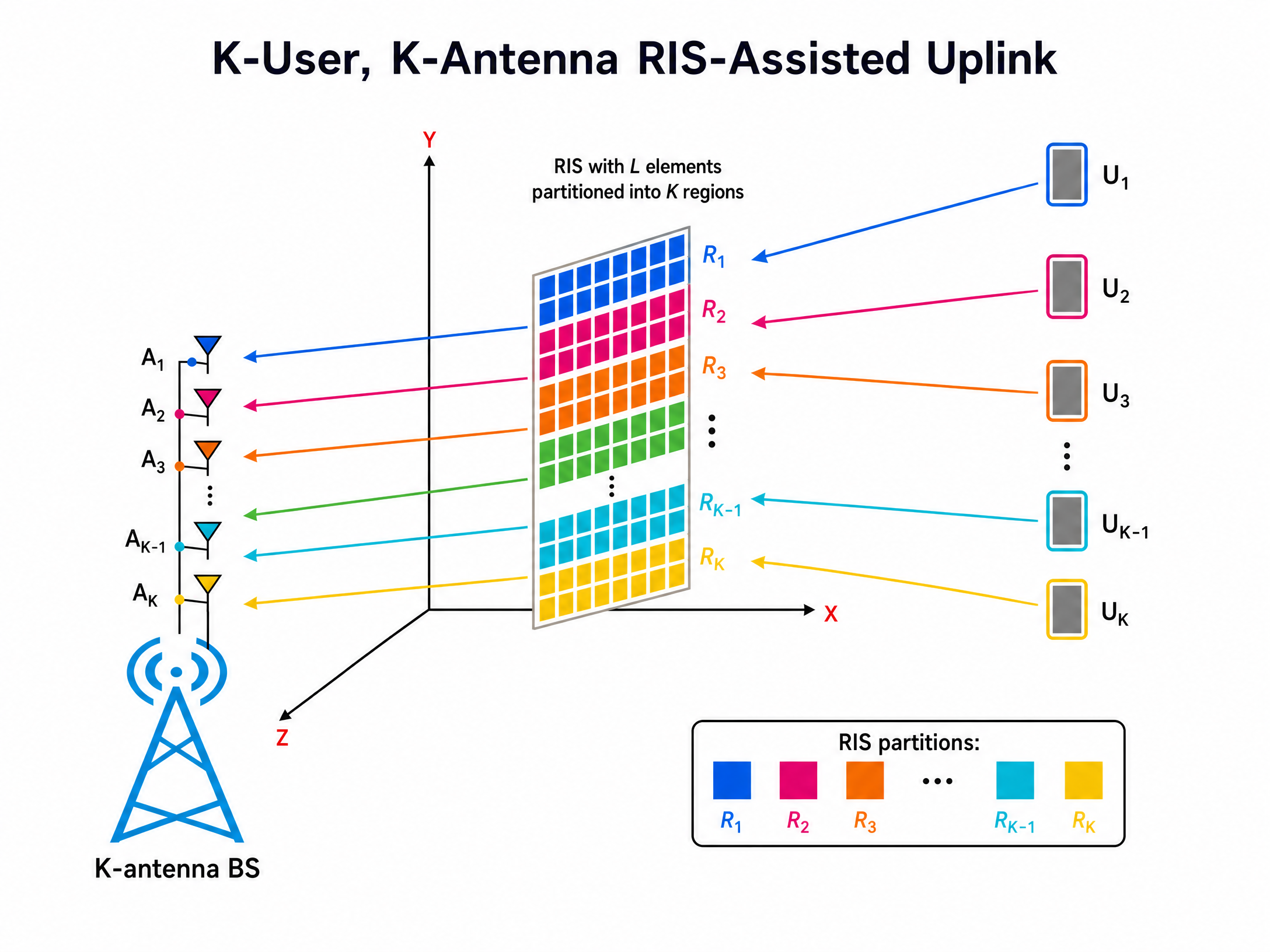}
		\caption{K-user, K-antenna system model with L RIS elements divided into K partitions.}
		\label{fig:Uplink_RIS_K_generalized_corrected}
	\end{figure}
	The $K$ receiver antennas form a uniform linear array (ULA) parallel to the
	$y$-axis, located at broadside distance $d_B$ from the RIS along the
	$x$-axis and centred in $y$ so that the array is symmetric with respect to
	the RIS. With inter-antenna spacing $d_A$, the position of the $i$-th
	antenna ($i=1,\ldots,K$) is
	\begin{equation}
		\mathbf{p}_{a,i} = \left(-d_B,\;
		\Bigl(\tfrac{K+1}{2}-i\Bigr)d_A,\; 0\right).
		\label{eq:antenna_pos}
	\end{equation}
	
	The channel between the $n$-th RIS element and the $i$-th receiver antenna is modeled as free-space line-of-sight (LOS):
	\begin{equation}
		h_{in} = \frac{1}{2\sqrt{\pi} r_{i,n}} e^{-j\frac{2\pi}{\lambda} r_{i,n}},
	\end{equation}
	where $r_{i,n}$ is the Euclidean distance between the $i$-th receive antenna and the $n$-th RIS element. The channel between User $k$ and the $n$-th RIS element is modeled as $\sqrt{\alpha_k} t_{nk}$, where $\alpha_k$ is the path-loss coefficient and $t_{nk}$ follows an independent Rayleigh fading distribution, $t_{nk} \sim \mathcal{CN}(0, 1)$.
	
	The $L \times K$ channel
	matrix between the users and RIS is thus
	\begin{equation}
		{\mathbf{T}} = \left[\sqrt{\alpha_1}\,\mathbf{t}_1\;\;\sqrt{\alpha_2}\,\mathbf{t}_2\;\;\cdots\;\;\sqrt{\alpha_K}\,\mathbf{t}_K\right],
		\label{eq:T_tilde}
	\end{equation}
	where $\mathbf{t}_k \in \mathbb{C}^{L\times1}$ includes the fading
	coefficients $t_{nk}$, $n=1,\ldots,L$, for user $k$. The information
	symbols of the $K$ users are represented by $x_1,\ldots,x_K$ and
	$\mathbf{x} = [x_1\;\cdots\;x_K]^T$ denotes the transmit signal vector.
	The transmit power of the users is denoted as $P$. The users transmit
	their signals simultaneously using the same frequency bandwidth. The
	user signals are refracted and reconfigured by the RIS, then reach the
	antennas at the receiver. The processed received signal vector, denoted
	as $\mathbf{y}=[y_1\;\cdots\;y_K]^T\in \mathbb{C}^{K \times 1}$, can thus be represented as the following.
	\begin{equation}
		\mathbf{y} = \sqrt{P}\,\mathbf{V}\mathbf{H}\boldsymbol{\Phi} {\mathbf{T}}\mathbf{x}
		+ \mathbf{V} {\mathbf{n}},
		\label{eq:y_matrix}
	\end{equation}
	where
	\begin{equation}
		\boldsymbol{\Phi} = \mathrm{diag}\!\left(e^{j\phi_1}, e^{j\phi_2}, \ldots, e^{j\phi_L}\right)
		\label{eq:Phi_matrix}
	\end{equation}
	is the RIS beamforming matrix with $e^{j\phi_n}$ being the phase-shift
	applied by the $n$-th RIS element, $\mathbf{H}\in\mathbb{C}^{K\times L}$ is
	the channel matrix from the RIS to the $K$ receive antennas (with
	entries $h_{in}$), $\mathbf{V}$ is the beamforming matrix at the
	receiver, and $ {\mathbf{n}}$ is the noise vector at the receiver
	antennas, following $\mathcal{CN}(\mathbf{0}_K,\mathbf{I}_K)$.  The effective channel matrix representing the channel from User $k$ to Antenna $i$ is denoted as
	$\mathbf{G}=\mathbf{H}\boldsymbol{\Phi} {\mathbf{T}}\in \mathbb{C}^{K \times K}$ with entries:
	
	\begin{equation}
		g_{ik} = \sqrt{\alpha_k} \sum_{n=1}^L h_{in} t_{nk} e^{j\phi_n},
	\end{equation}
	
	\section{Combined RIS Partition-Based Beamforming and ZF Detection}
	\label{sec:proof}
	The natural idea of RIS beamforming is to form coherent
	combination of the signals by canceling the phases of the
	propagation links passing through each RIS element. This
	is clearly the optimal design for communications with a
	single user. Unfortunately, for multi-user systems where user
	interference exist, such a design is far from
	straightforward. In this work, we consider the case where
	the number of RIS elements is very large, and propose to
	combine partition-based beamforming at the RIS and ZF
	beamforming at the receiver antennas, taking advantage of the
	strong capability of RIS in interference mitigation and system 
	performance.
	
	\subsection{RIS Phase Shift and Receiver Antenna Beamforming Designs}
	
	For the $K$-user system, we partition the RIS into $k$ subsurfaces,
	each assigned to a specific user, and responsible for
	the cancellation of its assigned user’s channel phases. The set
	of RIS elements of the $k$th sub-surface, serving User k, is
	denoted as 
	which are partitioned into $K$ sub-surfaces $R_1, R_2, \dots, R_K$, where the $k$-th sub-surface $R_k$ is dedicated to User $k$ and responsible for the beamforming of the $k$th user’s signal. The number of elements in each partition is denoted by $L_{pk} = |R_k|$, satisfying $\sum_{k=1}^K L_{pk} = L$.
	The partitions can be expressed as
	\begin{equation}
		\begin{split}
			& \mathcal{R}_1
			=
			\{1,\ldots,L_{p_1}\},
			\\
			&	\mathcal{R}_k
			=
			\left\{
			\sum_{m=1}^{k-1}L_{p_m}+1,
			\ldots,
			\sum_{m=1}^{k}L_{p_m}
			\right\} \text{for $k>1$}.
		\end{split}
		\label{eq:seq_partition}
	\end{equation} 
	Then, the partition-based RIS beamforming for elements in $R_k$ is designed to cancel the phase of the effective channel from User $k$ to Antenna $k$ as follows:
	\begin{equation}
		\phi_n = -\angle(h_{kn} t_{nk}), \quad \forall n \in R_k.
		\label{eq:phase_design}
	\end{equation}
	
	For ZF at the receiver antennas, the beamforming matrix is given as the following.
	\begin{equation}
		\mathbf{V}_{\text{ZF}} = \left(\mathbf{G}^H\mathbf{G}\right)^{-1}\mathbf{G}^H,
		\label{eq:zf_pseudo}
	\end{equation}
	where $[\cdot]^H$ represents the Hermitian. For
	the considered case of $K$ receive antennas and $K$ users the beamforming matrix in~\eqref{eq:zf_pseudo} reduces to $\mathbf{V}_{\text{ZF}}=\mathbf{G}^{-1}$. The RIS beamforming and the ZF
	beamforming are both in closed-form and their computational
	complexity is $\mathcal{O}(L)$.
	\subsection{Performance Analysis}
	In what follows, we analyze the user average achievable rate and
	the system sum-rate. With \eqref{eq:zf_pseudo}, the receive SINR of User k
	can be written as
	\begin{equation}
		\gamma_{\text{ZF},k} = \frac{P}{[(\mathbf{G}^H \mathbf{G})^{-1}]_{kk}}.
		\label{eq:zf_sinr_exact}
	\end{equation}
	where $[\mathbf{A}]_{ij}$ represents the $(i, j)$th entry of matrix $\mathbf{A}$.
	The term $[(\mathbf{G}^H\mathbf{G})^{-1}]_{kk}$ is the $k$-th diagonal element of the inverse Gram matrix. By using the matrix inversion lemma and properties of block/cofactor decomposition, we have:
	
	\begin{equation}
		\begin{split}
			\gamma_{\text{ZF},k}=	\frac{P}{[(\mathbf{G}^H\mathbf{G})^{-1}]_{kk}}
			= P\mathbf{g}_k^H \mathbf{P}_{-k}^{\perp} \mathbf{g}_k,
			\\
			\mathbf{P}_{-k}^{\perp}
			=
			\mathbf{I}
			-
			\mathbf{G}_{-k}
			\left(
			\mathbf{G}_{-k}^{H}\mathbf{G}_{-k}
			\right)^{-1}
			\mathbf{G}_{-k}^{H}.
			\label{eq:zf_sinr_projector}
		\end{split}
	\end{equation}

	Where $	\mathbf{P}_{-k}^{\perp}$ denote the orthogonal projector onto the complement of $\operatorname{span}(\mathbf{G}_{-k})$, where
	$\mathbf{G}_{-k}$ is obtained by removing the $k$-th column from
	$\mathbf{G}$. Also, $\mathbf{g}_k=[
	g_{1k}, g_{2k}, \cdots, g_{Kk}]^{T}$ is the effective channel vector corresponding to User~$k$.

	In the following, we want to find an asymptotic closed-form expression for the average achievable rates of users, $\bar{R}_k$. According to the law of large numbers, we have:
	\begin{equation}
		\bar{R}_k=\mathbb{E}[\log_{2}(1+	\gamma_{\text{ZF},k})]\approx\log_{2}(1+	\mathbb{E}[\gamma_{\text{ZF},k}])=\log_{2}(1+	  \bar{\gamma}_{\text{ZF},k})
		\label{eq:avg_rate}
	\end{equation}
	and $	\bar{\gamma}_{\text{ZF},k}=P\mathbb{E}[\mathbf{g}_k^H \mathbf{P}_{-k}^{\perp} \mathbf{g}_k].$
	
	In order to find a closed-form asymptotic expression for the average achievable rates of user, we start by the derivation of $\mathbb{E}\!\left[\mathbf{g}_k^{H}\mathbf{P}_{-k}^{\perp}\mathbf{g}_k\right]$.
	Using the idempotence of the projection matrix, $\mathbf{P}_{-k}^{\perp} = (\mathbf{P}_{-k}^{\perp})^2$, we obtain
	\[
	\mathbf{g}_k^{H}
	\mathbf{P}_{-k}^{\perp}
	\mathbf{g}_k
	=
	\left\|
	\mathbf{P}_{-k}^{\perp}
	\mathbf{g}_k
	\right\|^{2}
	\ge 0.
	\]
	Now, consider the following vector decomposition for $\mathbf{g}_k$.
	\begin{equation}
		\mathbf{g}_k
		=
		\sqrt{\alpha_k}
		\left(
		S_k\mathbf{e}_k
		+
		\boldsymbol{\zeta}_k
		\right),
	\end{equation}
	where $S_k
	=
	\sum_{n\in\mathcal{R}_k}
	|h_{kn}|\,|t_{nk}|
	$ is the coherent channel gain for User $k$, and $\mathbf{e}_k$ is the $k$-th canonical basis vector. Further, $\boldsymbol{\zeta}_k$ is the interference channel vector
	\[
	\boldsymbol{\zeta}_k
	=
	\mathbf{a}_k^{\perp}
	+
	\mathbf{b}_k,
	\] 
	where $\mathbf{a}_k^{\perp}$ is defined such that $[\mathbf{a}_k^{\perp}]_k = 0$, and for $i \neq k$,
	\begin{equation}
		\left[\mathbf{a}_k^{\perp}\right]_i
		= \sum_{n \in \mathcal{R}_k} |h_{in}|\,|t_{nk}|\,
		e^{j\left(\angle h_{in} - \angle h_{kn}\right)}.
		\label{eq:a_perp_entry}
	\end{equation}
	Also, $ [\mathbf{b}_k]_i
	=
	\sum_{n\notin\mathcal{R}_k}
	h_{in}t_{nk}e^{j\phi_n}.$
	From \eqref{eq:a_perp_entry}, since $|t_{nk}|$ is Rayleigh distributed with
	$\mathbb{E}[|t_{nk}|] = \sqrt{\pi}/2$, so
	\begin{equation}
		\left[\mathbb{E}[\boldsymbol{\zeta}_k]\right]_i
		= \frac{\sqrt{\pi}}{2}
		\sum_{n \in \mathcal{R}_k} |h_{in}|\,
		e^{j\left(\angle h_{in} - \angle h_{kn}\right)},
		\qquad i \neq k,
		\label{eq:mu_k_def}
	\end{equation}
	and $\left[\mathbb{E}[\boldsymbol{\zeta}_k]\right]_k = 0$. Substituting the decomposition above,
	\begin{equation}
		\begin{split}
			&	\left\|
			\mathbf{P}_{-k}^{\perp}
			\mathbf{g}_k
			\right\|^{2}
			=\\
			&
			\alpha_k
			\Big(
			S_k^{2}
			\|
			\mathbf{P}_{-k}^{\perp}\mathbf{e}_k
			\|^{2}
			+
			2S_k
			\operatorname{Re}
			\{
			\mathbf{e}_k^{H}
			\mathbf{P}_{-k}^{\perp}
			\boldsymbol{\zeta}_k
			\}
			+
			\|
			\mathbf{P}_{-k}^{\perp}
			\boldsymbol{\zeta}_k
			\|^{2}
			\Big).
		\end{split}
		\label{eq:quad_expansion}
	\end{equation}
	In the following, we will characterize $\mathbf{G}_{-k}$  and $\mathbf{P}_{-k}^{\perp}$ to further evaluate the expectation of each of the three terms above.
	
	Since $\mathbf{G}_{-k} \in \mathbb{C}^{K\times(K-1)}$ has full column rank almost surely
	(its $K-1$ columns are generically linearly independent for $K \leq L$), the rank-nullity theorem
	applied to $\mathbf{G}_{-k}^H : \mathbb{C}^{K} \to \mathbb{C}^{K-1}$ gives
	\begin{equation}
		\dim\!\left(\ker\!\left(\mathbf{G}_{-k}^H\right)\right)
		= K - \operatorname{rank}(\mathbf{G}_{-k}) = 1.
		\label{eq:null_dim_main}
	\end{equation}
	Let $\mathbf{u}_k \in \mathbb{C}^K$, $\|\mathbf{u}_k\|=1$, denote the unique (up to a global
	phase) unit vector satisfying $\mathbf{G}_{-k}^H\mathbf{u}_k = \mathbf{0}$.
	The orthogonal projector onto this null space is the rank-one matrix
	\begin{equation}
		\mathbf{P}_{-k}^{\perp} = \mathbf{u}_k\mathbf{u}_k^H,
		\label{eq:rank1_projector}
	\end{equation}
	so that
	\begin{equation}
		\|\mathbf{P}_{-k}^{\perp}\mathbf{e}_k\|^2
		= \mathbf{e}_k^H\mathbf{u}_k\mathbf{u}_k^H\mathbf{e}_k
		= |\mathbf{u}_k^H\mathbf{e}_k|^2
		= |[{\mathbf{u}_k}]_k|^2,
		\label{eq:proj_squared}
	\end{equation}
	where $[{\mathbf{u}_k}]_k$ denotes the $k$-th entry of $\mathbf{u}_k$. Therefore, evaluating
	\eqref{eq:proj_squared}
	requires characterizing the distribution of the null vector $\mathbf{u}_k$.
	\textcolor{black}{For this, we continue evaluating the mean and covariance of the elements of  $\mathbf{G}_{-k}$. The mean vector of the $\mathbf{g}_j$ $j \neq k$ columns of $\mathbf{G}_{-k}$ is shown by  $\mathbf{m}_j$.}
	We have
	\begin{equation}
		\begin{split}
			g_{ij} &= \sqrt{\alpha_j}\sum_{n=1}^{L}h_{in}t_{nj}e^{j\phi_n}, \qquad j \neq k
			\\
			&=\sqrt{\alpha_j}
			\underbrace{\sum_{n \in \mathcal{R}_j}h_{in}t_{nj}e^{j\phi_n}}_{A_{ij}}
			+\sqrt{\alpha_j}
			\underbrace{\sum_{n \notin \mathcal{R}_j}h_{in}t_{nj}e^{j\phi_n}}_{B_{ij}}, j \neq k.
			\\
			&\mathbb{E}[A_{ij}]
			= \frac{\sqrt{\pi}}{2}\sum_{n \in \mathcal{R}_j}|h_{in}|\,e^{j(\angle h_{in}-\angle h_{jn})}, \mathbb{E}[B_{ij}] = 0
			\label{eq:gij_def}
		\end{split}
	\end{equation}
	
	Therefore, the mean of the $(i,j)$th entry of $\mathbf{G}_{-k}$ is
	\begin{equation}
		\begin{split}
			[\mathbf{m}_j]_i= \mathbb{E}[g_{ij}] = \frac{\sqrt{\pi}}{2}\sqrt{\alpha_j}
			\sum_{n \in \mathcal{R}_j}|h_{in}|\,e^{j(\angle h_{in}-\angle h_{jn})}.
		\end{split}
		\label{eq:m_j_def}
	\end{equation}
	The mean vector $\mathbf{m}_j$ is non-zero for all $i$. In particular, the $j$-th (diagonal) entry
	\begin{equation}
		[\mathbf{m}_j]_j
		= \frac{\sqrt{\pi}}{2}\sqrt{\alpha_j}\sum_{n \in \mathcal{R}_j}|h_{jn}|
		= O(L_{pj})
		\label{eq:m_j_diag}
	\end{equation}
	is real, positive, and grows linearly with the partition size $L_{pj}$, while for $i \neq j$,
	\begin{equation}
		[\mathbf{m}_j]_i
		= \frac{\sqrt{\pi}}{2}\sqrt{\alpha_j}\sum_{n \in \mathcal{R}_j}|h_{in}|\,
		e^{-j\frac{2\pi}{\lambda}(r_{in}-r_{jn})}
		\label{eq:m_j_offdiag}
	\end{equation}
	is a sum of $L_{pj}$ complex phasors whose phases
	\begin{equation}
		\theta_n \triangleq -\frac{2\pi}{\lambda}(r_{in}-r_{jn}), \qquad n \in \mathcal{R}_j,
		\label{eq:phase_diff}
	\end{equation}
	depend on the path-length differences $r_{in}-r_{jn}$ between antenna $i$ and antenna $j$
	to each RIS element $n$. We now bound this sum for the confirmed antenna
	geometry of~\eqref{eq:antenna_pos}, under the following standard far-field
	assumption. Because $\mathcal{R}_j$ is
	defined as a contiguous block of flattened indices~\eqref{eq:seq_partition},
	it generally spans \emph{multiple columns} of the physical grid, so the phase progression across $\mathcal{R}_j$ must be
	treated as two-dimensional.
	\begin{assumption}[Far-field planar phase progression]
		\label{ass:farfield}
		The BS antennas $i,j$ are two elements of the ULA at the base station,
		with positions $\mathbf{p}_{a,i},\mathbf{p}_{a,j}$ as given
		in~\eqref{eq:antenna_pos}. The
		RIS is a planar $M\times M$ array with inter-element spacing $d$ (typically
		$d=\lambda/2$) in both the row and column directions. The BS array and the
		RIS are in each other's far field, so that the path-length difference from
		element $(m,p)$ satisfies exactly
		\begin{equation}
			r_{i,(m,p)} - r_{j,(m,p)} = \Delta_{ij} + m\,d\,\nu_{ij}^{(y)}
			+ p\,d\,\nu_{ij}^{(z)},
			\quad m,p = 0,\ldots,M-1,
			\label{eq:farfield}
		\end{equation}
		where $\Delta_{ij}$ is the path-length difference at the reference element
		$(0,0)$, and $\nu_{ij}^{(y)},\nu_{ij}^{(z)}\in(-2,2)$ are the normalised
		path differences along the row and column directions,
		\begin{equation}
			\nu_{ij}^{(y)} \triangleq \frac{\|\mathbf{p}_{a,i}-\mathbf{p}_{a,j}\|}{\|\mathbf{p}_{\mathrm{RIS}}\|}
			\sin\phi_{ij}^{(y)}, \qquad
			\nu_{ij}^{(z)} \triangleq \frac{\|\mathbf{p}_{a,i}-\mathbf{p}_{a,j}\|}{\|\mathbf{p}_{\mathrm{RIS}}\|}
			\sin\phi_{ij}^{(z)},
			\label{eq:nu_def}
		\end{equation}
		with $\phi_{ij}^{(y)},\phi_{ij}^{(z)}$ the angles of the antenna-separation
		vector $\mathbf{p}_{a,i}-\mathbf{p}_{a,j}$ (along the BS ULA) projected onto the RIS
		row and column axes respectively, and $\mathbf{p}_{\mathrm{RIS}}$ the
		position vector of either antenna relative to the RIS centre (the origin
		of Section~\ref{sec:system}), so that $\|\mathbf{p}_{\mathrm{RIS}}\|\approx
		d_B$ for antennas near broadside.
	\end{assumption}
	Assumption~\ref{ass:farfield} is the standard far-field (plane-wave) model
	used in array signal processing~\cite{vanTrees2002}. Under it, and using the
	column-major partition structure of~\eqref{eq:seq_partition},
	Appendix~\ref{app:projector} shows that the mean channel matrix
	$\mathbf{M}_{-k}\triangleq[\mathbf{m}_j]_{j\neq k}$ is diagonally dominant,
	with $|[\mathbf{m}_j]_j|=O(L_{pj})$ on the diagonal against
	$|[\mathbf{m}_j]_i|=O(L_{pj}/\sqrt{L})$ off it for $i\neq j$, and that the
	fluctuation-to-mean energy ratio vanishes as $L_{pj}\to\infty$, so that
	$\mathbf{G}_{-k}\to\mathbf{M}_{-k}$ in mean square. Building on this,
	Appendix~\ref{app:projector} further shows, in three steps, that the
	null-space projector converges entrywise to the rank-one limit
	$\mathbf{P}_{-k}^{\perp}=\mathbf{u}_k\mathbf{u}_k^H\xrightarrow{p}\mathbf{e}_k\mathbf{e}_k^H$
	(equation~\eqref{eq:Pperp_matrix_conv}), with, in particular,
	$\mathbb{E}[\|\mathbf{P}_{-k}^{\perp}\mathbf{e}_k\|^2]\to1$
	(equation~\eqref{eq:proj_one}). We apply this matrix-level convergence to the
	derivations of each term in \eqref{eq:quad_expansion} that follow.
	\subsection{Asymptotic Approximation of $\mathbb{E}\!\left[S_k^{2}\,
		\|\mathbf{P}_{-k}^{\perp}\mathbf{e}_k\|^{2}\right]$}

	Since $S_k \perp \mathbf{P}_{-k}^{\perp}$,
	\begin{equation}
		\mathbb{E}\!\left[S_k^2\,\|\mathbf{P}_{-k}^{\perp}\mathbf{e}_k\|^{2}\right]
		= \mathbb{E}[S_k^2]\,\mathbb{E}\!\left[\|\mathbf{P}_{-k}^{\perp}\mathbf{e}_k\|^2\right].
		\label{eq:factorisation}
	\end{equation}
	By \eqref{eq:proj_one}, the second factor satisfies
	\begin{equation}
		\mathbb{E}\!\left[\|\mathbf{P}_{-k}^{\perp}\mathbf{e}_k\|^2\right]
		\to 1
		\quad\text{as } L_{pj}\to\infty,
		\label{eq:Pperp_ek_approx}
	\end{equation}
	so the first term reduces to $\mathbb{E}[S_k^2]$, evaluated exactly in the following.
	
	\medskip
	\noindent\textbf{Evaluation of  {$\mathbb{E}[S_k^2]$} }
	\label{sec:Sk2}
	
	Since $t_{nk}\sim\mathcal{CN}(0,1)$, $|t_{nk}|$ follows a Rayleigh distribution with the same standard moments used in Section~\ref{sec:proof}~\cite[Ch.~2]{Kay1993}:
	\[
	\mathbb{E}[|t_{nk}|] = \frac{\sqrt{\pi}}{2}, \qquad \mathbb{E}[|t_{nk}|^{2}] = 1, \qquad
	\mathrm{Var}(|t_{nk}|) = 1-\frac{\pi}{4}.
	\]
	Using independence across RIS elements,
	\[
	\mathbb{E}[S_k^2]
	= \hspace{-3mm}\sum_{n \in \mathcal{R}_k}|h_{kn}|^2\,\mathbb{E}[|t_{nk}|^2]
	+\hspace{-3mm} \sum_{\substack{n,m \in \mathcal{R}_k \\ n \neq m}} |h_{kn}||h_{km}|\,\mathbb{E}[|t_{nk}|]\,\mathbb{E}[|t_{mk}|],
	\]
	which simplifies to
	\[
	\boxed{
		\mathbb{E}[S_k^{2}]
		=
		\left(
		1-\frac{\pi}{4}
		\right)
		\sum_{n\in\mathcal{R}_k}
		|h_{kn}|^{2}
		+
		\frac{\pi}{4}
		\left(
		\sum_{n\in\mathcal{R}_k}
		|h_{kn}|
		\right)^{2}.
	}
	\]

	\subsection{Asymptotic Approximation of  {$\mathbb{E}[\|\mathbf{P}_{-k}^{\perp}\boldsymbol{\zeta}_k\|^2]$} }\label{sec:pzeta}
	
	Substituting the matrix approximation $\mathbf{P}_{-k}^{\perp} \approx
	\mathbf{e}_k\mathbf{e}_k^H$ established in \eqref{eq:Pperp_matrix_conv},
	\begin{equation}
		\mathbb{E}\!\left[\|\mathbf{P}_{-k}^{\perp}\boldsymbol{\zeta}_k\|^2\right]
		\approx \mathbb{E}\!\left[\|\mathbf{e}_k\mathbf{e}_k^H\boldsymbol{\zeta}_k\|^2\right]
		= \mathbb{E}\!\left[|[\boldsymbol{\zeta}_k]_k|^2\right].
		\label{eq:leakage_approx_ek}
	\end{equation}
	
	Recall $\boldsymbol{\zeta}_k = \mathbf{a}_k^{\perp} + \mathbf{b}_k$ with $[\mathbf{a}_k^{\perp}]_k = 0$ by construction, so
	$[\boldsymbol{\zeta}_k]_k = [\mathbf{b}_k]_k
	= \sum_{n\notin\mathcal{R}_k}h_{kn}\,t_{nk}\,e^{j\phi_n}$,
	where $\phi_n$ is independent of $t_{nk}$ and $\mathbb{E}[|t_{nk}|^2]=1$, so
	\begin{align}
		\mathbb{E}\!\left[|[\boldsymbol{\zeta}_k]_k|^2\right]
		= \mathbb{E}\!\left[|[\mathbf{b}_k]_k|^2\right]
		&=\hspace{-3mm} \sum_{n \notin \mathcal{R}_k}|h_{kn}|^2\,\mathbb{E}[|t_{nk}|^2]
		= \hspace{-3mm}\sum_{n \notin \mathcal{R}_k}|h_{kn}|^2.
		\label{eq:bk_k_sq}
	\end{align}
	and, then
	\begin{equation}
		\mathbb{E}\!\left[\|\mathbf{P}_{-k}^{\perp}\boldsymbol{\zeta}_k\|^2\right]
		\approx \sum_{n\notin\mathcal{R}_k}|h_{kn}|^2
		\;\leq\; \sum_{n=1}^{L}|h_{kn}|^2
		\;=\; O(L).
		\label{eq:leakage_Pperp_approx}
	\end{equation}

	Since $\mathbb{E}[S_k^2] = O(L_{pk}^2)$ (Section~\ref{sec:Sk2}),  \eqref{eq:leakage_Pperp_approx}
	is asymptotically negligible:
	\begin{equation}
		\frac{\mathbb{E}\!\left[\|\mathbf{P}_{-k}^{\perp}\boldsymbol{\zeta}_k\|^2\right]}
		{\mathbb{E}[S_k^2]}
		= O\!\left(\frac{L}{L_{pk}^2}\right) \to 0
		\quad\text{as } L_{pk}\to\infty.
		\label{eq:leakage_negligible}
	\end{equation}

	\subsection{Asymptotic Approximation of The Cross Term, $2S_k
		\operatorname{Re}
		\{
		\mathbf{e}_k^{H}
		\mathbf{P}_{-k}^{\perp}
		\boldsymbol{\zeta}_k
		\}$}
	\label{sec:cross-term}
	
	The cross term in \eqref{eq:quad_expansion} is
	\begin{equation}
		2\,\mathbb{E}\!\left[S_k\,\mathrm{Re}\!\left(
		\mathbf{e}_k^H \mathbf{P}_{-k}^{\perp} \boldsymbol{\zeta}_k\right)\right],
		\label{eq:cross_term}
	\end{equation}
	where $S_k$, $\boldsymbol{\zeta}_k$, and $\mathbf{P}_{-k}^{\perp}$ are all random.
	
	The $i=k$ contribution to this cross term vanishes exactly, since
	$[\mathbf{a}_k^{\perp}]_k=0$ by the phase-shift design and
	$\mathbb{E}[S_k\mathbf{b}_k]=\mathbf{0}$ by independence. For $i\neq k$,
	Appendix~\ref{app:crossterm} bounds each contribution by separately
	rate-bounding $\mathbb{E}[u_{kk}\,u_{ki}^*]$ and
	$[\mathbb{E}[S_k\boldsymbol{\zeta}_k]]_i$, showing that each is
	$O(L_{pk})$ (equation~\eqref{eq:cross_term_order}).
	We will now compare the order of the cross term relative to the leading-order signal term, 	$\mathbb{E}[S_k^2]=O(L_{pk}^2)=O(L^2/K^2)$ (Section~\ref{sec:Sk2}) to find the final approximation.  Since
	there are at most $K-1$ terms in the sum (fixed $K$), summing
	\eqref{eq:cross_term_order} over $i\neq k$ preserves the order, and
	\begin{equation}
		\begin{split}
			\frac{2\,\mathbb{E}\!\left[S_k\,\mathrm{Re}\!\left(
				\mathbf{e}_k^H \mathbf{P}_{-k}^{\perp} \boldsymbol{\zeta}_k\right)\right]}
			{\mathbb{E}[S_k^2]}
			&= O\!\left(\frac{L/K^{3/2}}{L^2/K^2}\right)\\
			&= O\!\left(\frac{\sqrt{K}}{L}\right)
			\;\xrightarrow{L\to\infty}\; 0
			\quad\text{for fixed } K.
		\end{split}
		\label{eq:cross_zero}
	\end{equation}
	This means that the cross term vanishes at rate $O(1/L)$ which makes it asymptotically negligible.
	The $i=k$ contribution to the cross term is \emph{exactly} zero due to the
	phase-shift design ($[\mathbf{a}_k^{\perp}]_k=0$ by construction and
	$\mathbb{E}[S_k\mathbf{b}_k]=\mathbf{0}$ by independence). The $i\neq k$
	contributions, from~\eqref{eq:cross_term_order} expressed in terms of
	$L_{pk}$ alone (using $L=KL_{pk}$, $K$ fixed): each term is
	$O(L_{pk}^{3/2}/\sqrt{L}) = O(L_{pk}^{3/2}/\sqrt{KL_{pk}}) = O(L_{pk}/\sqrt{K})
	= O(L_{pk})$, and summing the (at most $K-1$, fixed) terms preserves this
	order, so they are jointly $O(L_{pk})$ --- not individually vanishing in
	absolute terms, but negligible compared with the leading $O(L_{pk}^2)$
	signal term $\mathbb{E}[S_k^2]$, as shown in
	\eqref{eq:cross_zero}. This relative negligibility becomes arbitrarily strong
	as the partition size $L_{pj}\to\infty$, which is the regime of interest in
	this work.

	\subsection{Final Expression for  {$\mathbb{E}[\mathbf{g}_k^H\mathbf{P}_{-k}^{\perp}\mathbf{g}_k]$} }
	
	From~\eqref{eq:quad_expansion}, using the relative negligibility of the cross
	term and {$\mathbb{E}[\|\mathbf{P}_{-k}^{\perp}\boldsymbol{\zeta}_k\|^2]$} term, ,$	2S_k
	\operatorname{Re}
	\{
	\mathbf{e}_k^{H}
	\mathbf{P}_{-k}^{\perp}
	\boldsymbol{\zeta}_k
	\}$ established in Sections~\ref{sec:cross-term}, and \ref{sec:pzeta}, we obtain the following closed-form asymptotic
	approximation valid as $L_{pj}\to\infty$:
	\begin{equation}
		\begin{split}
			\mathbb{E}\!\left[\mathbf{g}_k^H \mathbf{P}_{-k}^{\perp} \mathbf{g}_k\right]
			&\approx
			\alpha_k\left\{\mathbb{E}\!\left[S_k^{2}\,
			\|\mathbf{P}_{-k}^{\perp}\mathbf{e}_k\|^{2}\right]\right\}=\alpha_k\,\mathbb{E}[S_k^2]\\
			&
			=	\alpha_k\!\left[
			\left(1-\frac{\pi}{4}\right)\sum_{n\in\mathcal{R}_k}|h_{kn}|^{2}
			+\frac{\pi}{4}\!\left(\sum_{n\in\mathcal{R}_k}|h_{kn}|\right)^{\!2}
			\right].
			\label{eq:final_two_term}
		\end{split}
	\end{equation}

	\section{Asymptotic ZF SINR}
	\label{sec:asymptotic}
	Recalling 
	from~\eqref{eq:zf_sinr_projector}, equation~\eqref{eq:final_two_term}
	establishes the closed-form average asymptotic ZF SINR stated in
	Theorem~\ref{thm:main} below.
	
	We summarise the $K$-user uplink system of Section~\ref{sec:system}: $K$
	single-antenna users are served by $K$ receive antennas assisted by an
	$L$-element RIS partitioned into contiguous sub-surfaces
	$\{\mathcal{R}_k\}$~\eqref{eq:seq_partition}, with per-element RIS phase
	design~\eqref{eq:phase_design} and ZF receive beamforming ~\eqref{eq:zf_pseudo}. Under
	Assumption~\ref{ass:farfield} and the LOS channel between the RIS and BS antennas in which every
	partition size $L_{pk}\to\infty$, the derivation of
	Section~\ref{sec:proof} yields the following closed-form result.
	
	\begin{theorem}[Closed-form asymptotic average ZF SINR]
		\label{thm:main}
		As $L_{pk}\to\infty$ for all $k$, the average ZF SINR of User $k$ satisfies
		\begin{equation}
			\boxed{
				\bar{\gamma }_{\mathrm{ZF},k} \approx P\alpha_k\!\left[
				\Bigl(1-\tfrac{\pi}{4}\Bigr)\!\!\sum_{n\in\mathcal{R}_k}\!\!|h_{kn}|^2
				+\tfrac{\pi}{4}\Bigl(\sum_{n\in\mathcal{R}_k}\!|h_{kn}|\Bigr)^{\!2}
				\right]. }
			\label{eq:final_strong_los}
		\end{equation}
	\end{theorem}
	\begin{proof}[Proof]
		Full derivations are given in
		Section~\ref{sec:proof} of this work.
	\end{proof}
	
	\subsection{Sum-Rate Optimality of Equal Partitioning}
	\label{sec:optimality}
	
	Theorem~\ref{thm:main} lets us characterize which
	partition sizes $\{L_{pk}\}$ maximize the sum-rate
	$R_{\mathrm{sum}}=\sum_k\log_2(1+\bar{\gamma}_{\mathrm{ZF},k})$ at leading order. More specifically, The coherent term in~\eqref{eq:final_strong_los}
	dominates the incoherent term, since it is $O(L_{pk}^2)$ against
	$O(L_{pk})$. Approximating each antenna's LOS amplitude as roughly constant
	over its own sub-surface, $|h_{kn}|\approx\bar{h}_k$ for $n\in\mathcal{R}_k$, reasonable since $\mathcal{R}_k$ is a physically localized block of
	the array, the dominant term simplifies to
	$\bigl(\sum_{n\in\mathcal{R}_k}|h_{kn}|\bigr)^2\approx
	\bar{h}_k^2L_{pk}^2$, giving
	\begin{equation}
		\bar{\gamma}_{\mathrm{ZF},k} \approx c_k L_{pk}^2,
		\qquad c_k \triangleq \frac{\pi}{4}P\alpha_k\bar{h}_k^2,
		\label{eq:gamma_ck}
	\end{equation}
	where $c_k$ depends on the path loss $\alpha_k$ and channel strength
	$\bar{h}_k$ but not on the partition sizes. For $L_{pk}$ large enough
	that $c_kL_{pk}^2\gg1$, $\log_2(1+c_kL_{pk}^2)\approx\log_2(c_kL_{pk}^2)$,
	so
	\begin{equation}
		R_{\mathrm{sum}} \approx \sum_{k=1}^{K}\log_2(c_k)
		+ 2\sum_{k=1}^{K}\log_2(L_{pk}).
		\label{eq:rsum_approx}
	\end{equation}
	The first term is a constant, independent of $\{L_{pk}\}$, and thus the partition
	only affects $R_{\mathrm{sum}}$ through $\sum_k\log_2(L_{pk})$.
	
	\begin{corollary}[Sum-rate optimality of equal partitioning]
		\label{cor:equal}
		Subject to $\sum_{k=1}^K L_{pk}=L$, the approximate sum-rate
		in~\eqref{eq:rsum_approx} is maximized by equal partitioning,
		$L_{pk}^*=L/K$ for all $k$, independently of $\{\alpha_k\}$ and
		$\{\bar{h}_k\}$.
	\end{corollary}
	\begin{proof}
		Maximize $f(\{L_{pk}\})=\sum_k\log_2(L_{pk})$ subject to
		$\sum_kL_{pk}=L$. The Lagrangian is
		$f(\{L_{pk}\})-\mu\bigl(\sum_kL_{pk}-L\bigr)$; the stationarity
		condition $\partial/\partial L_{pk}=1/(L_{pk}\ln2)-\mu=0$ gives
		$L_{pk}=1/(\mu\ln2)$ for every $k$, i.e.\ all $L_{pk}$ equal, and the
		constraint then fixes $L_{pk}^*=L/K$. The Hessian of $f$ is
		$\mathrm{diag}(-1/(L_{pk}^2\ln2))$, negative definite, so this
		stationary point is the unique maximizer. Since $c_k$ does not
		appear in $f$, the maximizer does not depend on
		$\{\alpha_k\},\{\bar{h}_k\}$.
	\end{proof}
	
	Corollary~\ref{cor:equal} is an approximate, leading-order result: it
	relies on the high-SINR approximation
	$\log_2(1+c_kL_{pk}^2)\approx\log_2(c_kL_{pk}^2)$ and on treating
	$|h_{kn}|$ as constant within each sub-surface, whereas the true LOS
	amplitudes vary across the array (a coefficient of variation of
	roughly $18$--$19\%$ within a partition in our simulation setting). It
	therefore predicts equal partitioning as the leading-order optimum
	rather than guaranteeing it exactly; Section~\ref{sec:sim} confirms this
	empirically via the greedy pairwise-transfer search, which operates on
	the closed-form objective~\eqref{eq:final_strong_los}
	and consistently converges to partitions mildly perturbed from
	$L_{pk}=L/K$, rather than to exactly equal partitions.

	\section{Simulation Results}
	\label{sec:sim}
	
	We use $\lambda=0.1$~m, $d_A=2\lambda$, $d_B=5\lambda$, RIS element
	spacing $\lambda/2$, and $1000$ Monte Carlo realizations per
	configuration. For each realization, the exact ZF SINR is computed via
	the equation~\eqref{eq:zf_sinr_exact}, $\gamma_{\mathrm{ZF},k} =
	P/[(\mathbf{G}^H\mathbf{G})^{-1}]_{kk}$, using the full random effective
	channel matrix $\mathbf{G}$; the simulated rate is
	$\log_2(1+\gamma_{\mathrm{ZF},k})$, averaged over realizations. The
	asymptotic rate uses the closed-form average SINR of Theorem~\ref{thm:main}
	(equation~\eqref{eq:final_strong_los}) directly in
	$\log_2(1+\bar{\gamma}_{\mathrm{ZF},k})$.
	
	\subsection{Validation of the Average Asymptotic Rate}
	
	Fig.~\ref{fig:sumrate} compares the simulated (solid) and
	asymptotic~\eqref{eq:final_strong_los} (dashed) sum-rate for
	$K\in\{2,4,8\}$ and $L\in\{1024,4096\}$ under equal partitioning and
	$\alpha_k=1$. The gap between the two, averaged over the SNR sweep,
	decreases with $L$ for every $K$: $0.21\%\to0.19\%$ for $K=2$,
	$0.98\%\to0.47\%$ for $K=4$, and $11.86\%\to2.80\%$ for $K=8$, confirming
	Theorem~\ref{thm:main} as an accurate predictor once the RIS is
	sufficiently large relative to $K$.
	
	\begin{figure}[t]
		\centering
		\includegraphics[width=\linewidth]{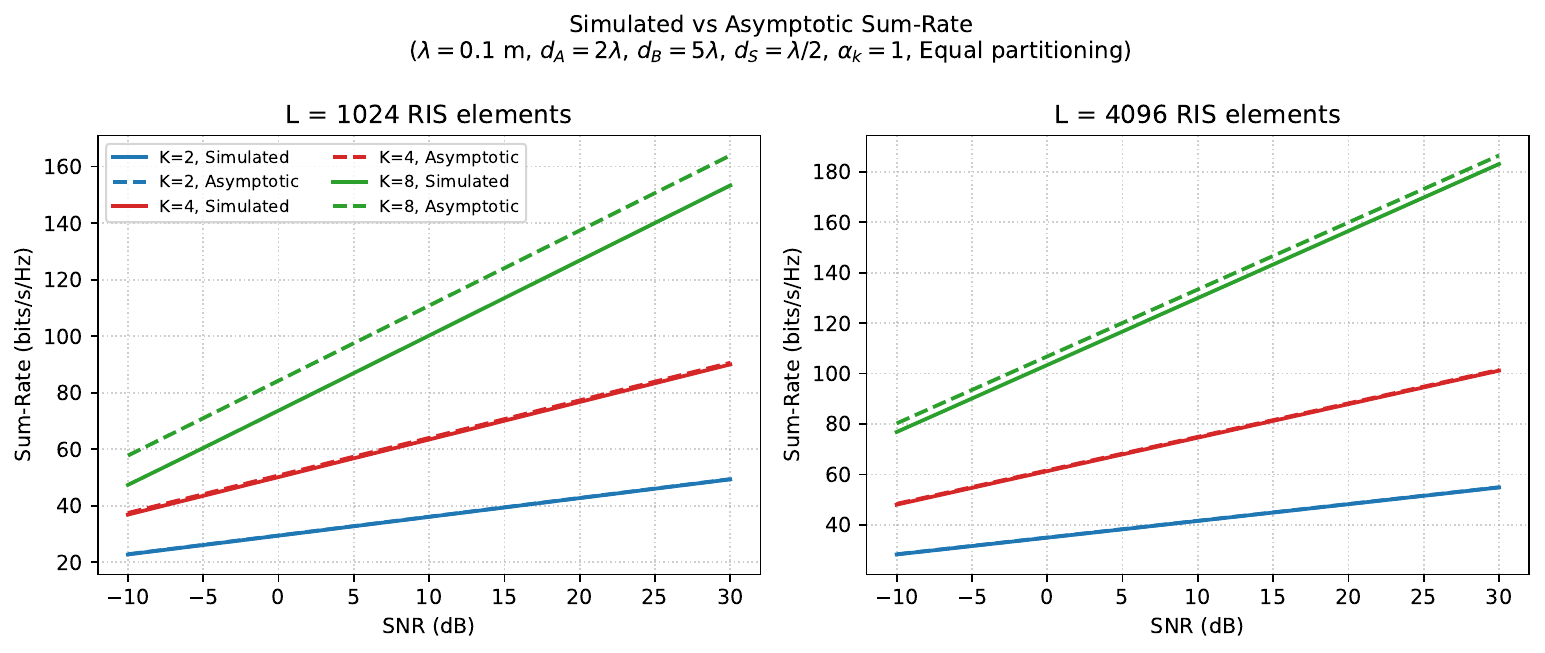}
		\caption{Sum-rate vs.\ SNR, $K\in\{2,4,8\}$, $L\in\{1024,4096\}$, equal
			partitioning, $\alpha_k=1$. Solid: Monte Carlo (exact SINR,
			equation~\eqref{eq:zf_sinr_exact}); dashed: asymptotic
			(Theorem~\ref{thm:main}).}
		\label{fig:sumrate}
	\end{figure}
	
	\subsection{Sum-Rate Maximisation: Search vs.\ Equal Partitioning}
	Corollary~\ref{cor:equal} identifies the leading-order optimal partition
	under simplifying approximations; here we state and address the asymptotic
	discrete optimization problem those approximations bypass. Using the
	closed-form average asymptotic rate $\log_2(1+\bar{\gamma}_{\mathrm{ZF},k})$ from
	Theorem~\ref{thm:main}~\eqref{eq:final_strong_los} as the per-user
	objective, the sum-rate-maximizing partition solves
	\begin{equation}
		\begin{split}
			\max_{\{L_{pk}\}_{k=1}^{K}} \quad
			& R_{\mathrm{sum}}(\{L_{pk}\})
			= \sum_{k=1}^{K}\log_2\!\bigl(1+\bar{\gamma}_{\mathrm{ZF},k}(L_{pk})\bigr) \\
			\text{s.t.} \quad
			& \sum_{k=1}^{K}L_{pk} = L,  L_{pk}\in\mathbb{Z}_{>0}, \quad k=1,\ldots,K,
		\end{split}
		\label{eq:sumrate_problem}
	\end{equation}
	where the partitions $\{\mathcal{R}_k\}$ induced by $\{L_{pk}\}$ remain
	contiguous blocks of flattened RIS indices as
	in~\eqref{eq:seq_partition}, and $\bar{\gamma}_{\mathrm{ZF},k}(L_{pk})$ is
	given by~\eqref{eq:final_strong_los}. Problem~\eqref{eq:sumrate_problem}
	is a discrete resource-allocation problem over the $K-1$ free integer
	partition boundaries; $R_{\mathrm{sum}}$ is neither convex nor concave in
	$\{L_{pk}\}$ once the asymptotic per-user rates are used, and
	an exhaustive search over all $\binom{L-1}{K-1}$ compositions of $L$ into
	$K$ positive parts is computationally infeasible for the RIS sizes
	considered in this work. Using~\eqref{eq:final_strong_los} as a fast, deterministic objective, the
	greedy pairwise-transfer search finds the
	sum-rate-maximizing partition without any channel realizations: starting
	from equal partitioning, it repeatedly transfers one element between the
	pair of users that most increases the asymptotic sum-rate, until no
	single-element transfer improves it. We then validate the result the
	search finds against real Monte Carlo simulation for asymmetric large scale scenarios.
	
	For $K=4$, $L=1024$, $\alpha=[2,1,0.5,0.1]$, $P=15$~dB, the search
	converges after $64$ iterations to sizes $[288,224,224,288]$, a small
	perturbation away from equal partitioning $[256,256,256,256]$; we
	verified explicitly that no single-element transfer from this partition
	improves the objective further, confirming a true local optimum.
	Fig.~\ref{fig:sumratemax_K4} confirms this improvement holds under real
	simulated channels: the search-optimal partition achieves
	$66.80$~bits/s/Hz versus $66.73$~bits/s/Hz for equal partitioning
	($+0.108\%$), closely tracking the asymptotic prediction of
	$+0.144\%$ improvement.
	
	\begin{figure}[t]
		\centering
		\includegraphics[width=\linewidth]{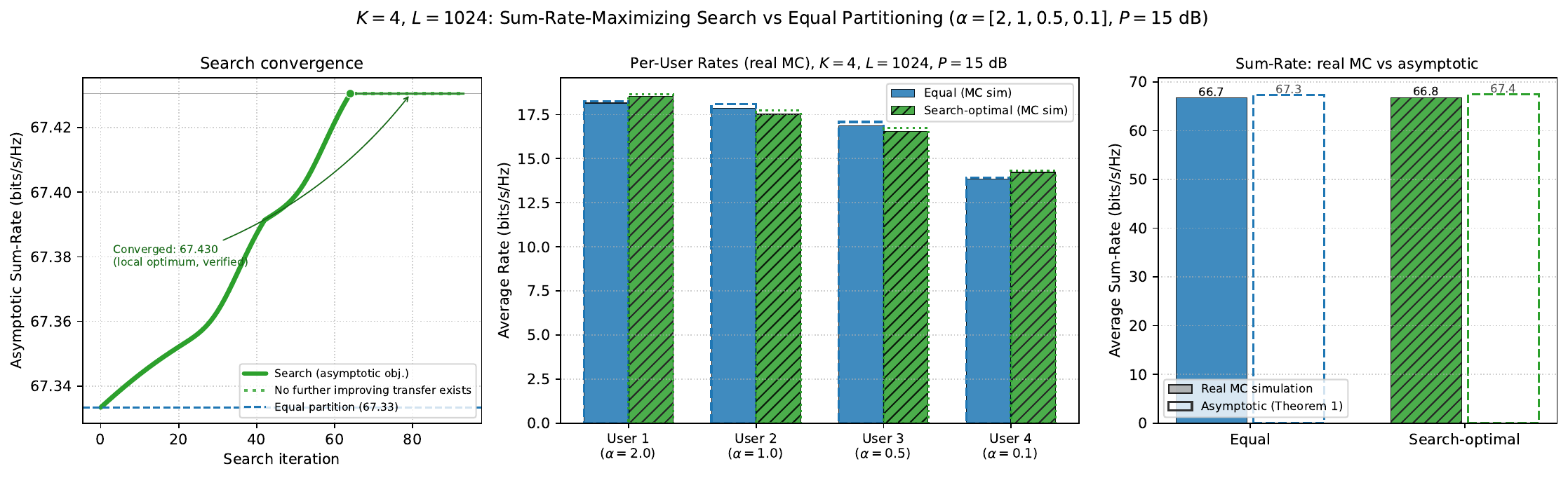}
		\caption{$K=4$, $L=1024$, $\alpha=[2,1,0.5,0.1]$, $P=15$~dB: search
			convergence (left), per-user rates from real Monte Carlo simulation
			(centre), and sum-rate comparison against the asymptotic objective
			(right).}
		\label{fig:sumratemax_K4}
	\end{figure}
	
	For $K=8$, $L=4096$, $\alpha=[2,1.5,1,0.8,0.5,0.3,0.2,0.1]$, $P=20$~dB,
	the search converges after $276$ iterations to
	$[640,522,453,433,433,453,522,640]$, again a verified local optimum with
	no further improving single-element transfer. Fig.~\ref{fig:sumratemax_K8} shows
	the same conclusion at this larger scale: the search-optimal partition
	achieves $150.15$~bits/s/Hz versus $149.41$~bits/s/Hz for equal
	partitioning under real simulation ($+0.492\%$), again consistent in
	direction and magnitude with the asymptotic prediction ($+0.353\%$).
	
	\begin{figure}[t]
		\centering
		\includegraphics[width=\linewidth]{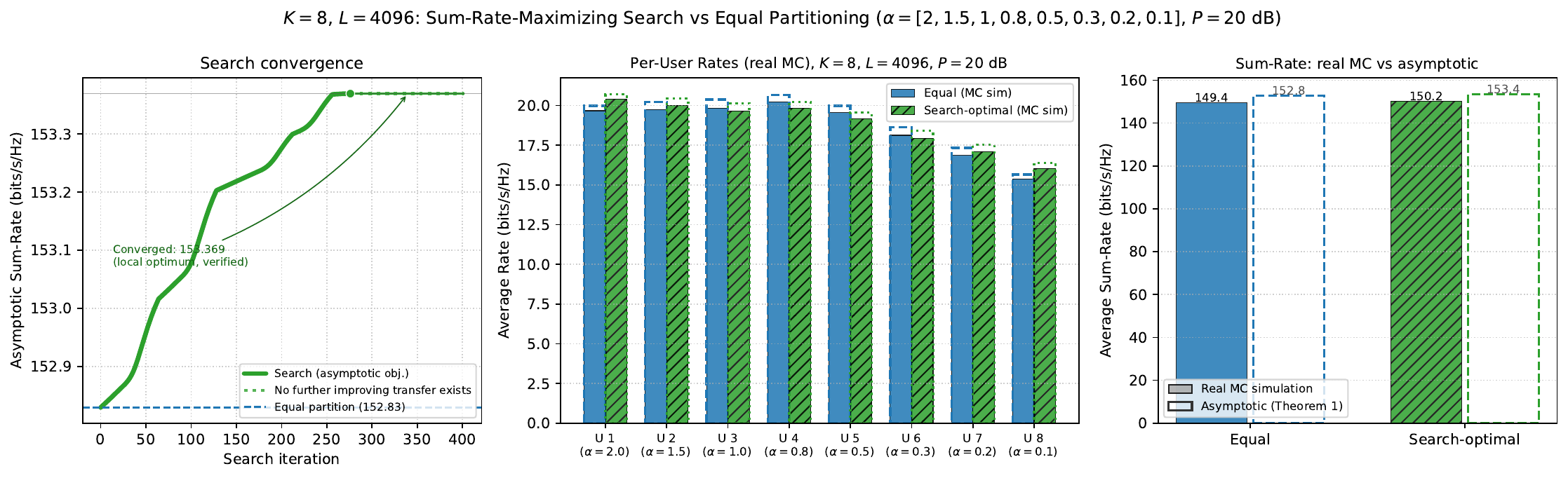}
		\caption{$K=8$, $L=4096$, $\alpha=[2,1.5,1,0.8,0.5,0.3,0.2,0.1]$,
			$P=20$~dB: search convergence (left), per-user rates from real
			Monte Carlo simulation (centre), and sum-rate comparison against the
			asymptotic objective (right).}
		\label{fig:sumratemax_K8}
	\end{figure}
	
	In both cases, the improvement the search finds is small  which is consistent with Corollary~\ref{cor:equal}, which shows equal
	partitioning is approximately sum-rate-optimal at leading order. But
	it is a genuine, positive improvement confirmed by real simulated
	channels: the search-optimal partition never simulates worse
	than equal partitioning in either scenario, and the amount of improvement under real simulation is always within the same order as the
	asymptotic prediction.
	
	\section{Conclusion}
	\label{sec:conclusion}
	We derived a closed-form asymptotic approximation, obtained through a large-scale analysis in which the RIS sub-surface sizes grow without bound, for the average achievable rates of users in a large-scale RIS-equipped uplink system under partition-based RIS beamforming and ZF receive beamforming. Simulation results confirm the accuracy of this closed-form expression across a range of system sizes, with the required RIS size for a given accuracy growing with the number of users. Building on this expression, we further showed analytically that equal partitioning is approximately sum-rate-optimal at leading order, independently of path-loss asymmetry (Corollary~\ref{cor:equal}), and validated a greedy pairwise-transfer search against real Monte Carlo simulation, confirming that the improvements it finds over equal partitioning are genuine rather than artifacts of the asymptotic objective the search itself relies on. Future work includes jointly optimizing the partition sizes and phase-shift design under a max-min-optimal achievable-rate criterion, and extending the partition-based RIS beamforming to quantized phase-shifting at the RIS.

	\appendices
	\section{Proof of the Asymptotic Rank-One Projector Approximation}
	\label{app:projector}
	
	Assumption~\ref{ass:farfield} is the standard far-field (plane-wave) model
	used in array signal processing~\cite{vanTrees2002}. Under it, the
	phase of element $(m,p)$ satisfies \emph{exactly}
	\begin{equation}
		\theta_{(m,p)} = \theta_{(0,0)} + 2\pi m\alpha_{ij}^{(y)} + 2\pi p\alpha_{ij}^{(z)}, \alpha_{ij}^{(y)}\triangleq\frac{d}{\lambda}\nu_{ij}^{(y)}, \alpha_{ij}^{(z)}\triangleq\frac{d}{\lambda}\nu_{ij}^{(z)}.
		\label{eq:phase_arith}
	\end{equation}
	Given Assumption~\ref{ass:farfield}, we now bound
	\eqref{eq:m_j_offdiag} directly for the confirmed antenna geometry
	of~\eqref{eq:antenna_pos}: every BS antenna lies at $z=0$, so for any pair
	$i\neq j$ the antenna-separation vector $\mathbf{p}_{a,i}-\mathbf{p}_{a,j}$
	has zero $z$-component, giving $\nu_{ij}^{(z)}=0$ and hence
	$\alpha_{ij}^{(z)}=0\in\mathbb{Z}$ \emph{identically} for this array. The
	$y$-direction remains generically non-resonant, since antennas differ in
	$y$ by~\eqref{eq:antenna_pos} and $\alpha_{ij}^{(y)}$ is irrational for
	generic $d_A/\lambda$.
	
	\smallskip
	\noindent\textbf{Bounding the off-diagonal mean: three-piece decomposition.}
	Since $\mathcal{R}_j$ is a contiguous block of the column-major flattened index
	(equation~\eqref{eq:seq_partition}), it decomposes into at most three pieces:
	a partial first column $\mathcal{R}_j^{(1)}$ (rows $m=m_1,\ldots,M-1$ of a
	single column $p_0$), a block of full columns
	$\mathcal{R}_j^{(2)} = \{0,\ldots,M-1\}\times\{p_0+1,\ldots,p_0+K_j\}$, and
	a partial last column $\mathcal{R}_j^{(3)}$ (rows $m=0,\ldots,m_2$ of column
	$p_0+K_j+1$), where $K_j = \lfloor L_{pj}/M \rfloor$ is the number of full
	columns. Bounding $|h_{in}|\leq c/r_{\min}$ (LOS amplitude bounded uniformly by
	the minimum element-to-antenna distance $r_{\min}>0$), each piece is bounded
	separately.
	
	\emph{Full-columns block $\mathcal{R}_j^{(2)}$.} The double sum factors into a
	product of a row sum and a column sum,
	\begin{equation}
		\begin{split}
			&\left|\sum_{(m,p)\in\mathcal{R}_j^{(2)}}\! e^{j\theta_{(m,p)}}\right|
			= \left|\sum_{m=0}^{M-1}\! e^{j2\pi m\alpha_{ij}^{(y)}}\right|
			\cdot
			\left|\sum_{p=p_0+1}^{p_0+K_j} e^{j2\pi p\cdot 0}\right|.
		\end{split}
		\label{eq:geom_series_block_setup}
	\end{equation}
	The $y$-direction factor is a \emph{full} sweep over all $M$ rows with
	$\alpha_{ij}^{(y)}\notin\mathbb{Z}$, bounded by the standard geometric-series
	bound $1/|\sin(\pi\alpha_{ij}^{(y)})| = O(1)$, independent of $K_j$. The
	$z$-direction factor, since $\alpha_{ij}^{(z)}=0$ identically, is not an
	oscillating sum at all: $\bigl|\sum_{p=p_0+1}^{p_0+K_j}e^{j2\pi p\cdot0}\bigr|
	= K_j$ exactly, with no cancellation. Hence
	\begin{equation}
		\Bigl|\sum_{(m,p)\in\mathcal{R}_j^{(2)}}e^{j\theta_{(m,p)}}\Bigr|
		= O(1)\cdot K_j = O(K_j) =  O\!\left(\frac{L_{pj}}{\sqrt{L}}\right),
		\label{eq:offdiag_colmajor}
	\end{equation}
	using $K_j = \lfloor L_{pj}/M\rfloor$ and $M=\sqrt{L}$.
	
	\emph{Partial columns $\mathcal{R}_j^{(1)}, \mathcal{R}_j^{(3)}$.} Each is a
	single-column sum over a contiguous range of rows --- purely a $y$-direction
	sum, unaffected by the $z$-direction degeneracy --- itself a geometric
	series bounded by
	\begin{equation}
		\left|\sum_{(m,p)\in\mathcal{R}_j^{(1)}} e^{j\theta_{(m,p)}}\right|
		\leq \frac{1}{|\sin(\pi\alpha_{ij}^{(y)})|} = O(1),
		\label{eq:geom_series_partial}
	\end{equation}
	and analogously for $\mathcal{R}_j^{(3)}$, independent of $L_{pj}$.
	
	\emph{Combining the three pieces.} By the triangle inequality,
	\begin{equation}
		\begin{split}
			\left|\sum_{n\in\mathcal{R}_j} e^{j\theta_n}\right|
			\leq \sum_{i=1}^{3}\left|\sum_{n\in\mathcal{R}_j^{(i)}} e^{j\theta_n}\right|
			&= O\!\left(\frac{L_{pj}}{\sqrt{L}}\right) + O(1) + O(1)\\
			&= O\!\left(\frac{L_{pj}}{\sqrt{L}}\right),
			\label{eq:geom_series}
		\end{split}
	\end{equation}
	the full-columns block dominating the two partial-column terms. Substituting
	into \eqref{eq:m_j_offdiag},
	\begin{equation}
		|[\mathbf{m}_j]_i|
		\leq \frac{\sqrt{\pi}}{2}\sqrt{\alpha_j}\,\frac{c}{r_{\min}}\cdot
		O\!\left(\frac{L_{pj}}{\sqrt{L}}\right)
		= O\!\left(\frac{L_{pj}}{\sqrt{L}}\right).
		\label{eq:offdiag_O1}
	\end{equation}
	Since the diagonal entry satisfies $[\mathbf{m}_j]_j = O(L_{pj})$
	(equation~\eqref{eq:m_j_diag}), we obtain
	\begin{equation}
		\boxed{
			\frac{|[\mathbf{m}_j]_i|}{|[\mathbf{m}_j]_j|}
			= \frac{O(L_{pj}/\sqrt{L})}{O(L_{pj})}
			= O\!\left(\frac{1}{\sqrt{L}}\right) \to 0
			\quad\hspace{-2mm}\text{as } L\to\infty, \hspace{-4mm}
			\quad i \neq j,
		}
		\label{eq:ratio_to_zero}
	\end{equation}
	provided $\alpha_{ij}^{(y)}\notin\mathbb{Z}$ --- the only non-resonance
	condition remaining, since $\alpha_{ij}^{(z)}=0$ is a fixed consequence of
	the array geometry~\eqref{eq:antenna_pos} rather than a genericity
	assumption. This still
	vanishes and preserves every qualitative conclusion of this work, at the
	rate $O(1/\sqrt{L})$ set by the array's broadside orientation, rather than
	the faster $O(1/L_{pj})$ rate that would hold if the antennas were not
	confined to a single plane.
	
	\smallskip
	\noindent\textbf{Non-resonance condition.}
	The condition $\alpha_{ij}^{(y)}\notin\mathbb{Z}$, required for
	\eqref{eq:geom_series_block_setup}--\eqref{eq:ratio_to_zero} to hold, fails
	only if the $y$-direction spatial frequency is an exact integer. Since
	$i,j$ are BS antenna indices on the ULA, this requires the antenna
	separation $\|\mathbf{p}_{a,i}-\mathbf{p}_{a,j}\|$ and the corresponding
	projection angle to satisfy an exact resonance condition. For the standard
	half-wavelength spacing $d=\lambda/2$, this requires $\nu_{ij}^{(y)}\in\{0,\pm
	2\}$, corresponding to zero antenna separation ($i=j$, excluded by
	definition) or grazing incidence. In all other configurations, i.e., for
	generic BS antenna spacing $d_A$, \eqref{eq:ratio_to_zero} holds and the
	off-diagonal entries of $\mathbf{M}_{-k}$ are negligible compared to the
	diagonal entries for large $L_{pj}$.
	
	
	\noindent\textbf{Covariance matrix $\mathbf{C}_j$.}
	Write $\mathbf{g}_j = \mathbf{m}_j + \mathbf{f}_j$ where $\mathbf{f}_j \triangleq \mathbf{g}_j - \mathbf{m}_j$
	is zero mean. The covariance matrix $\mathbf{C}_j = \mathbb{E}[\mathbf{f}_j\mathbf{f}_j^H]$ is computed
	by separating on-partition and off-partition contributions.
	
	For $n \in \mathcal{R}_j$, using \eqref{eq:phase_design}, the fluctuation in $\sqrt{\alpha_j}A_{ij}$ is
	\begin{equation}
		[f_j^A]_i = \sqrt{\alpha_j}\sum_{n \in \mathcal{R}_j}|h_{in}|
		\bigl(|t_{nj}|-\mathbb{E}[|t_{nj}|]\bigr)e^{j(\angle h_{in}-\angle h_{jn})},
	\end{equation}
	and the $(i,i')$ entry of the on-partition covariance is, using independence of $|t_{nj}|$ across $n$,
	\begin{equation}
		\begin{split}
			\left[\mathbb{E}[\mathbf{f}_j^A(\mathbf{f}_j^A)^H]\right]_{ii'}
			&= \alpha_j\,\mathrm{Var}(|t_{nj}|)\sum_{n \in \mathcal{R}_j}
			|h_{in}|\,|h_{i'n}|\,e^{j(\angle h_{in}-\angle h_{i'n})}\\
			&= \alpha_j\!\left(1-\frac{\pi}{4}\right)\sum_{n \in \mathcal{R}_j}[\mathbf{p}_n\mathbf{p}_n^H]_{ii'},
			\label{eq:Cj_on}
		\end{split}
	\end{equation}
	where for each $n\in\mathcal{R}_j$, $\mathbf{p}_n$
	is the vector with $i$-th entry $|h_{in}|\,e^{j(\angle h_{in}-\angle h_{jn})}$.
	
	For $n \notin \mathcal{R}_j$, since $t_{nj} \sim \mathcal{CN}(0,1)$ and $\phi_n \perp t_{nj}$, the
	fluctuation $[f_j^B]_i = \sqrt{\alpha_j}\sum_{n \notin \mathcal{R}_j}h_{in}t_{nj}e^{j\phi_n}$ has covariance
	\begin{equation}
		\begin{split}
			\left[\mathbb{E}[\mathbf{f}_j^B(\mathbf{f}_j^B)^H]\right]_{ii'}
			&	= \alpha_j\sum_{n \notin \mathcal{R}_j}h_{in}h_{i'n}^*\,\mathbb{E}[|t_{nj}|^2]
			= \alpha_j\sum_{n \notin \mathcal{R}_j}h_{in}h_{i'n}^*\\
			&	= \alpha_j\sum_{n \notin \mathcal{R}_j}[\mathbf{h}_n\mathbf{h}_n^H]_{ii'},
			\label{eq:Cj_off}
		\end{split}
	\end{equation}
	The cross terms between
	on-partition and off-partition fluctuations vanish since those sets of RIS elements are disjoint
	and $t_{nj}$ is independent across $n$. Summing \eqref{eq:Cj_on} and \eqref{eq:Cj_off},
	\begin{equation}
		\mathbf{C}_j = \mathbb{E}[\mathbf{f}_j\mathbf{f}_j^H]
		= \alpha_j\!\left(1-\frac{\pi}{4}\right)\sum_{n \in \mathcal{R}_j}\mathbf{p}_n\mathbf{p}_n^H
		+ \alpha_j\sum_{n \notin \mathcal{R}_j}\mathbf{h}_n\mathbf{h}_n^H.
		\label{eq:Cj_full}
	\end{equation}
	The full second moment of column $j$ is therefore
	\begin{equation}
		\mathbb{E}[\mathbf{g}_j\mathbf{g}_j^H] = \mathbf{m}_j\mathbf{m}_j^H + \mathbf{C}_j,
		\label{eq:second_moment_corrected}
	\end{equation}
	which is in general neither proportional to the identity nor a rank-one matrix.
	
	Define the mean matrix $\mathbf{M}_{-k} \triangleq [\mathbf{m}_j]_{j \neq k} \in \mathbb{C}^{K \times (K-1)}$. In the following, we establish the diagonal dominance of the mean matrix $\mathbf{M}_{-k}$.
	From \eqref{eq:m_j_diag} and \eqref{eq:offdiag_O1}, the $j$-th column of $\mathbf{M}_{-k}$ satisfies
	\begin{equation}
		|[\mathbf{m}_j]_j| = O(L_{pj}), \qquad
		|[\mathbf{m}_j]_i| = O\!\left(\frac{L_{pj}}{\sqrt{L}}\right) \text{ for } i \neq j,
		\label{eq:diag_dom}
	\end{equation}
	where $|\cdot|$ denotes the complex modulus. Hence $\mathbf{M}_{-k}$ is diagonally dominant in the sense that each column's
	mean magnitude is concentrated on its own diagonal index. We next compare the energy of $\mathbf{m}_j$ to
	that of the fluctuation $\mathbf{f}_j$. From \eqref{eq:m_j_diag},
	\begin{equation}
		\|\mathbf{m}_j\| \geq |[\mathbf{m}_j]_j|
		= \frac{\sqrt{\pi}}{2}\sqrt{\alpha_j}\sum_{n \in \mathcal{R}_j}|h_{jn}|
		= O(L_{pj}).
		\label{eq:m_norm}
	\end{equation}
	From \eqref{eq:Cj_full}, $\mathrm{tr}(\mathbf{C}_j) = \mathbb{E}[\|\mathbf{f}_j\|^2]$, and
	\begin{equation}
		\begin{split}
			\mathbb{E}[\|\mathbf{f}_j\|^2]
			&= \alpha_j\!\left(1-\frac{\pi}{4}\right)\sum_{n \in \mathcal{R}_j}\sum_i|h_{in}|^2\\
			&+ \alpha_j\sum_{n \notin \mathcal{R}_j}\sum_i|h_{in}|^2
			\leq \alpha_j\sum_{n=1}^{L}\sum_{i=1}^{K}|h_{in}|^2
			= O(L),
			\label{eq:f_norm}
		\end{split}
	\end{equation}
	so $\mathbb{E}[\|\mathbf{f}_j\|^2] \leq \alpha_j C_h L$ for some finite constant
	$C_h = \max_{i,n}|h_{in}|^2 \cdot K < \infty$ that depends only on the LOS geometry.
	From \eqref{eq:m_norm}, $\|\mathbf{m}_j\|^2 \geq |[\mathbf{m}_j]_j|^2 =
	\frac{\pi}{4}\alpha_j\bigl(\sum_{n\in\mathcal{R}_j}|h_{jn}|\bigr)^2 \geq
	\frac{\pi}{4}\alpha_j c_h^2 L_{pj}^2$, where $c_h = \min_{n\in\mathcal{R}_j}|h_{jn}| > 0$.
	Therefore the ratio of mean-square fluctuation energy to squared mean norm satisfies
	\begin{equation}
		\frac{\mathbb{E}[\|\mathbf{f}_j\|^2]}{\|\mathbf{m}_j\|^2}
		\leq \frac{\alpha_j C_h L}{\tfrac{\pi}{4}\alpha_j c_h^2 L_{pj}^2}
		= \frac{4C_h}{\pi c_h^2}\cdot\frac{L}{L_{pj}^2}
		\to 0 \quad\text{as } L_{pj}\to\infty,
		\label{eq:snr_ratio}
	\end{equation}
	since $L/L_{pj}^2 \to 0$ when $L_{pj}\to\infty$ (the partition grows without bound while
	$L/L_{pj}$ remains bounded).
	
	This gives a precise meaning to the approximation $\mathbf{g}_j\approx\mathbf{m}_j$:
	defining the normalised residual $\widetilde{\mathbf{f}}_j \triangleq
	\mathbf{f}_j/\|\mathbf{m}_j\|$, equation~\eqref{eq:snr_ratio} states
	\begin{equation}
		\mathbb{E}\!\left[\|\widetilde{\mathbf{f}}_j\|^2\right]
		= \frac{\mathbb{E}[\|\mathbf{f}_j\|^2]}{\|\mathbf{m}_j\|^2}
		\to 0
		\quad\text{as }L_{pj}\to\infty,
		\label{eq:ms_convergence}
	\end{equation}
	i.e., $\mathbf{f}_j/\|\mathbf{m}_j\| \to \mathbf{0}$ in mean square. We use convergence
	in mean square (m.s.) throughout, which implies convergence in probability. Equivalently,
	writing $\mathbf{g}_j = \mathbf{m}_j(\mathbf{1} + \widetilde{\mathbf{f}}_j)$, the
	relative fluctuation $\widetilde{\mathbf{f}}_j$ converges to zero in mean square,
	so $\mathbf{g}_j/\|\mathbf{m}_j\|\to\mathbf{m}_j/\|\mathbf{m}_j\|$ in mean square.
	Consequently $\mathbf{G}_{-k}\to\mathbf{M}_{-k}$ in the sense that each column
	converges to its mean in mean square relative to the mean norm.
	
	We now establish that $|[\mathbf{u}_k]_k|^2 \to 1$ in a probabilistic sense
	as $L_{pj}\to\infty$, proceeding in three steps.
	
	\medskip
	\noindent\textbf{Step 1:} In this step, we will show the convergence of $\mathbf{G}_{-k}$ to $\mathbf{M}_{-k}$
	which implies that the null-space equation concentrates.
	The null vector $\mathbf{u}_k$ satisfies $\mathbf{G}_{-k}^H\mathbf{u}_k = \mathbf{0}$,
	which can be written as
	\begin{equation}
		\mathbf{M}_{-k}^H\mathbf{u}_k + \mathbf{F}_{-k}^H\mathbf{u}_k = \mathbf{0},
		\label{eq:null_split}
	\end{equation}
	where $\mathbf{F}_{-k} \triangleq \mathbf{G}_{-k} - \mathbf{M}_{-k}$ is the zero-mean
	fluctuation matrix, so that, column by column, $\mathbf{m}_j^H\mathbf{u}_k =
	-\mathbf{f}_j^H\mathbf{u}_k$ for each $j\neq k$. Since $\|\mathbf{u}_k\|=1$, the
	Cauchy--Schwarz inequality gives
	\begin{equation}
		\mathbb{E}\!\left[|\mathbf{f}_j^H\mathbf{u}_k|^2\right]
		\leq \mathbb{E}\!\left[\|\mathbf{f}_j\|^2\right]\cdot\|\mathbf{u}_k\|^2
		= \mathbb{E}\!\left[\|\mathbf{f}_j\|^2\right],
		\label{eq:residual_bound}
	\end{equation}
	and dividing both sides by $\|\mathbf{m}_j\|^2$ and applying the ratio
	result~\eqref{eq:ms_convergence},
	\begin{equation}
		\frac{\mathbb{E}\!\left[|\mathbf{f}_j^H\mathbf{u}_k|^2\right]}{\|\mathbf{m}_j\|^2}
		\leq \frac{\mathbb{E}\!\left[\|\mathbf{f}_j\|^2\right]}{\|\mathbf{m}_j\|^2}
		= \mathbb{E}\!\left[\|\widetilde{\mathbf{f}}_j\|^2\right]
		\to 0
		\quad\text{as } L_{pj}\to\infty.
		\label{eq:residual_bound_normalised}
	\end{equation}
	Hence $\mathbf{f}_j^H\mathbf{u}_k/\|\mathbf{m}_j\| \xrightarrow{\mathrm{m.s.}} 0$,
	and since $\mathbf{m}_j^H\mathbf{u}_k = -\mathbf{f}_j^H\mathbf{u}_k$, the same holds
	for $\mathbf{m}_j^H\mathbf{u}_k$:
	\begin{equation}
		\begin{split}
			&\frac{\mathbf{m}_j^H\mathbf{u}_k}{\|\mathbf{m}_j\|}
			\xrightarrow{\mathrm{m.s.}} 0,
			\qquad	\text{equivalently}\\
			&
			\mathbf{m}_j^H\mathbf{u}_k = o_p(\|\mathbf{m}_j\|) = o_p(L_{pj})
			\quad\text{as }L_{pj}\to\infty,
			\label{eq:null_cond}
		\end{split}
	\end{equation}
	using $\|\mathbf{m}_j\| = O(L_{pj})$, which follows from \eqref{eq:diag_dom}: since
	$K$ is fixed, $\|\mathbf{m}_j\|^2 = |[\mathbf{m}_j]_j|^2 + \sum_{i\neq j}|[\mathbf{m}_j]_i|^2
	= O(L_{pj}^2) + (K-1)\cdot O(L_{pj}^2/L) = O(L_{pj}^2)$, so $\|\mathbf{m}_j\|=O(L_{pj})$;
	this is the upper bound needed here, complementing the lower bound
	$\|\mathbf{m}_j\|\geq|[\mathbf{m}_j]_j|$ already noted in \eqref{eq:m_norm}. 
	
	\medskip
	\noindent\textbf{Step 2:} In this step, we will show that  $[\mathbf{u}_k]_j \xrightarrow{p} 0$ for all $j\neq k$.
	By \eqref{eq:diag_dom}, $\|\mathbf{m}_j\|^2 = |[\mathbf{m}_j]_j|^2 +
	\sum_{i\neq j}|[\mathbf{m}_j]_i|^2 = |[\mathbf{m}_j]_j|^2\bigl(1+o(1)\bigr)$ as
	$L_{pj}\to\infty$, since the $K-1$ off-diagonal terms are each $O(L_{pj}/\sqrt{L})$
	while $|[\mathbf{m}_j]_j|=O(L_{pj})$ and $K$ is fixed, so
	$\|\mathbf{m}_j\|/|[\mathbf{m}_j]_j| \to 1$. Combining this with
	\eqref{eq:null_cond} gives the equivalent normalization
	\begin{equation}
		\frac{\mathbf{m}_j^H\mathbf{u}_k}{|[\mathbf{m}_j]_j|}
		\xrightarrow{p} 0
		\quad\text{as } L_{pj}\to\infty.
		\label{eq:null_cond_normalised}
	\end{equation}
	Expanding using \eqref{eq:diag_dom},
	\begin{equation}
		\mathbf{m}_j^H\mathbf{u}_k
		= [\mathbf{m}_j]_j^*[\mathbf{u}_k]_j
		+ \sum_{i \neq j}[\mathbf{m}_j]_i^*[\mathbf{u}_k]_i.
		\label{eq:null_expand}
	\end{equation}
	Since $\|\mathbf{u}_k\|=1$, every entry satisfies $|[\mathbf{u}_k]_i|\leq 1$, so
	\begin{equation}
		\left|\sum_{i\neq j}[\mathbf{m}_j]_i^*[\mathbf{u}_k]_i\right|
		\leq \sum_{i\neq j}|[\mathbf{m}_j]_i|
		\leq (K-1)\max_{i\neq j}|[\mathbf{m}_j]_i|
		= O\!\left(\frac{L_{pj}}{\sqrt{L}}\right),
		\label{eq:offdiag_small}
	\end{equation}
	using \eqref{eq:diag_dom} ($|[\mathbf{m}_j]_i|=O(L_{pj}/\sqrt{L})$ for $i\neq j$, and $K$ is
	fixed). The diagonal term satisfies $|[\mathbf{m}_j]_j|\geq \frac{\sqrt{\pi}}{2}
	\sqrt{\alpha_j}c_h L_{pj}$ from \eqref{eq:m_j_diag}. Therefore, rearranging
	\eqref{eq:null_expand} and dividing by $|[\mathbf{m}_j]_j|$,
	\begin{equation}
		|[\mathbf{u}_k]_j|
		= \frac{1}{|[\mathbf{m}_j]_j|}
		\left|\mathbf{m}_j^H\mathbf{u}_k
		- \sum_{i\neq j}[\mathbf{m}_j]_i^*[\mathbf{u}_k]_i\right|
		\leq \frac{|\mathbf{m}_j^H\mathbf{u}_k|}{|[\mathbf{m}_j]_j|}
		+ \frac{O(L_{pj}/\sqrt{L})}{|[\mathbf{m}_j]_j|}.
		\label{eq:uk_j_bound}
	\end{equation}
	The first term converges to zero in probability by
	\eqref{eq:null_cond_normalised}. The second term is
	$O(L_{pj}/\sqrt{L})/O(L_{pj}) = O(1/\sqrt{L})\to 0$ deterministically. Hence
	\begin{equation}
		[\mathbf{u}_k]_j \xrightarrow{p} 0
		\quad\text{as }L_{pj}\to\infty,
		\qquad \forall\,j\neq k.
		\label{eq:uk_j_to_zero}
	\end{equation}
	
	\medskip
	\noindent\textbf{Step 3:} In this step, we will show that  $\mathbb{E}[|[\mathbf{u}_k]_k|^2]\to 1$.
	Since $\|\mathbf{u}_k\|^2=1$ identically,
	\begin{equation}
		|[\mathbf{u}_k]_k|^2
		= 1 - \sum_{j\neq k}|[\mathbf{u}_k]_j|^2.
		\label{eq:uk_k_identity}
	\end{equation}
	By \eqref{eq:uk_j_to_zero}, $[\mathbf{u}_k]_j\xrightarrow{p}0$ for each $j\neq k$.
	We upgrade this to $\sum_{j\neq k}\mathbb{E}\!\left[|[\mathbf{u}_k]_j|^2\right]
	\to 0
	\quad\text{as }L_{pj}\to\infty$ in two steps: first to convergence in probability of the
	squared magnitude, then to convergence of its expectation.
	
	\smallskip
	\noindent\emph{(a) Convergence in probability of $|[\mathbf{u}_k]_j|^2$.}
	Let $g(x)\triangleq|x|^2$ (here $x$ is a dummy argument used to
	name the function $g$, not a symbol used elsewhere in this work); $g$ is
	continuous on $\mathbb{C}$. By the continuous mapping
	theorem~\cite[Thm.~25.7]{Billingsley1995}, $[\mathbf{u}_k]_j\xrightarrow{p}0$
	and $g$ continuous together give
	\begin{equation}
		|[\mathbf{u}_k]_j|^2 = g\bigl([\mathbf{u}_k]_j\bigr) \xrightarrow{p} g(0) = 0,
		\qquad j\neq k.
		\label{eq:cmt_step}
	\end{equation}
	
	\smallskip
	\noindent\emph{(b) Convergence of the expectation.} Convergence in
	probability alone does \emph{not} imply convergence of the expectation in
	general (a rare but large enough outlier can keep the expectation from
	following); an additional boundedness condition is needed. Since
	$\|\mathbf{u}_k\|=1$ exactly, every entry satisfies
	$0\leq|[\mathbf{u}_k]_j|^2\leq 1$, so this boundedness holds uniformly, with
	no dependence on $L_{pj}$. Applying the dominated convergence
	theorem~\cite[Thm.~16.4]{Billingsley1995} to~\eqref{eq:cmt_step},
	\begin{equation}
		\mathbb{E}\!\left[\sum_{j\neq k}|[\mathbf{u}_k]_j|^2\right]
		= \sum_{j\neq k}\mathbb{E}\!\left[|[\mathbf{u}_k]_j|^2\right]
		\to 0
		\quad\text{as }L_{pj}\to\infty.
		\label{eq:sum_to_zero}
	\end{equation}
	Taking expectations in \eqref{eq:uk_k_identity} and applying \eqref{eq:sum_to_zero},
	\begin{equation}
		\begin{split}
			&	\mathbb{E}\!\left[\|\mathbf{P}_{-k}^{\perp}\mathbf{e}_k\|^2\right]
			= \mathbb{E}\!\left[|[\mathbf{u}_k]_k|^2\right]\\
			&= 1 - \sum_{j\neq k}\mathbb{E}\!\left[|[\mathbf{u}_k]_j|^2\right]
			\;\to\; 1
			\quad\text{as }L_{pj}\to\infty.
			\label{eq:proj_one}
		\end{split}	
	\end{equation}
	The convergence holds in the sense that the expectation of the bounded random variable
	$|[\mathbf{u}_k]_k|^2\in[0,1]$ converges to $1$, which (since the variable is bounded
	above by $1$) also implies $|[\mathbf{u}_k]_k|^2\xrightarrow{p}1$. In the regime of
	this work, $L_{pj}=O(L/K)\to\infty$, so \eqref{eq:proj_one} is the relevant
	approximation.
	
	\medskip
	Since $\mathbf{P}_{-k}^{\perp}=\mathbf{u}_k\mathbf{u}_k^H$, its $(i,j)$ entry is
	$u_{ki}u_{kj}^*$. Using $|u_{ki}|\leq 1$ for all $i$ (as $\|\mathbf{u}_k\|=1$)
	together with \eqref{eq:uk_j_to_zero}, the off-diagonal entries satisfy, for
	$i\neq k$,
	\begin{equation}
		|u_{kk}u_{ki}^*| \leq |u_{ki}| \xrightarrow{p} 0,
		\qquad
		|u_{ki}u_{kk}^*| \leq |u_{ki}| \xrightarrow{p} 0,
		\label{eq:offdiag_matrix_conv}
	\end{equation}
	and for $i,j\neq k$,
	\begin{equation}
		|u_{ki}u_{kj}^*| \leq |u_{ki}|\,|u_{kj}| \xrightarrow{p} 0.
		\label{eq:offdiag_matrix_conv2}
	\end{equation}
	Together with $|u_{kk}|^2\xrightarrow{p}1$ from \eqref{eq:proj_one}, every entry of
	$\mathbf{u}_k\mathbf{u}_k^H$ converges in probability to the corresponding entry of
	$\mathbf{e}_k\mathbf{e}_k^H$:
	\begin{equation}
		\boxed{
			\mathbf{P}_{-k}^{\perp} = \mathbf{u}_k\mathbf{u}_k^H
			\xrightarrow{p} \mathbf{e}_k\mathbf{e}_k^H
			\quad\text{entrywise, as } L_{pj}\to\infty.
		}
		\label{eq:Pperp_matrix_conv}
	\end{equation}
	We will apply this matrix-level convergence
	$\mathbf{P}_{-k}^{\perp}\approx\mathbf{e}_k\mathbf{e}_k^H$ to the
	derivations of each term in \eqref{eq:quad_expansion} that follow.
	
	\section{Derivation of the Cross-Term Asymptotic Negligibility}
	\label{app:crossterm}
	
	\subsubsection{Sources of randomness}
	
	\begin{itemize}
		\item $S_k$ depends only on the fading channel magnitudes of user $k$, $\{|t_{nk}|\}_{n \in \mathcal{R}_k}$.
		\item $\boldsymbol{\zeta}_k = \mathbf{a}_k^{\perp} + \mathbf{b}_k$. The vector
		$\mathbf{a}_k^{\perp}$ depends only on $\{|t_{nk}|\}_{n \in \mathcal{R}_k}$. The vector $\mathbf{b}_k$ involves $t_{nk}$ for $n \notin \mathcal{R}_k$,
		but each of its entries is multiplied by a phase shift $\phi_n = -\angle(h_{jn}t_{nj})$
		designed for another user $j \neq k$, so $\mathbf{b}_k$ also carries a randomness from
		$\angle t_{nj}$ for $j \neq k$.
		\item $\mathbf{P}_{-k}^{\perp}$ depends on the columns $\mathbf{g}_j$, $j \neq k$, i.e., on
		$t_{nj}$ for all users $j \neq k$.
	\end{itemize}
	We now prove that $\mathbf{P}_{-k}^{\perp}$ is independent of the pair
	$(S_k,\boldsymbol{\zeta}_k)$. Define the two disjoint sets of fading channel variables
	\begin{equation}
		\mathcal{T}_k \;\triangleq\; \{t_{nk} : n = 1,\ldots,L\}, \mathcal{T}_{-k} \;\triangleq\; \{t_{nj} : n = 1,\ldots,L,\; j \neq k\},
	\end{equation}
	which are mutually independent by assumption. Observe that
	\begin{equation}
		S_k = f_S\!\left(\mathcal{T}_{k}\right),
		\qquad
		\mathbf{P}_{-k}^{\perp} = f_P(\mathcal{T}_{-k}),
	\end{equation}
	for measurable functions $f_S$ and $f_P$. Since $\mathcal{T}_k$ and $\mathcal{T}_{-k}$ are
	independent, $S_k \perp \mathbf{P}_{-k}^{\perp}$.
	
	For $\boldsymbol{\zeta}_k = \mathbf{a}_k^{\perp} + \mathbf{b}_k$, the component $\mathbf{a}_k^{\perp}$
	satisfies
	\begin{equation}
		\left[\mathbf{a}_k^{\perp}\right]_i
		= \sum_{n \in \mathcal{R}_k}|h_{in}|\,|t_{nk}|\,
		e^{j(\angle h_{in} - \angle h_{kn})},
		\qquad i \neq k,
	\end{equation}
	which is a function of $\mathcal{T}_k$ only, hence $\mathbf{a}_k^{\perp} \perp \mathbf{P}_{-k}^{\perp}$.
	
	The component $\mathbf{b}_k$ requires a separate argument. Its $i$-th entry is
	\begin{equation}
		[\mathbf{b}_k]_i
		= \sum_{n \notin \mathcal{R}_k} h_{in}\,t_{nk}\,e^{j\phi_n}, \phi_n = -\angle(h_{jn}\,t_{nj}), n \in \mathcal{R}_j,\; j \neq k,
	\end{equation}
	which involves variables from both $\mathcal{T}_k$ and $\mathcal{T}_{-k}$. To establish
	independence from $\mathbf{P}_{-k}^{\perp} = f_P(\mathcal{T}_{-k})$, we examine the conditional
	distribution of $[\mathbf{b}_k]_i$ given $\mathcal{T}_{-k}$. Conditioning on $\mathcal{T}_{-k}$
	renders each phase shift $\phi_n$ a fixed constant, so
	\begin{equation}
		[\mathbf{b}_k]_i \,\big|\, \mathcal{T}_{-k}
		= \sum_{n \notin \mathcal{R}_k} c_n\, t_{nk},
		\qquad c_n \triangleq h_{in}\,e^{j\phi_n} \in \mathbb{C} \text{ fixed},
	\end{equation}
	which is a linear combination of i.i.d.\ $\mathcal{CN}(0,1)$ variables with fixed coefficients.
	This conditional distribution does not depend on the specific realization of $\mathcal{T}_{-k}$
	used for conditioning, because $t_{nk} \sim \mathcal{CN}(0,1)$, and the $\mathcal{CN}(0,1)$
	distribution is invariant under multiplication by any fixed unit-modulus scalar~\cite{tse2005fundamentals}. Therefore
	\begin{equation}
		\mathcal{L}\!\left([\mathbf{b}_k]_i \,\big|\, \mathcal{T}_{-k} = \tau\right)
		= \mathcal{L}\!\left([\mathbf{b}_k]_i\right)
		\quad \forall\, \tau,
		\label{eq:bk_cond_dist}
	\end{equation}
	where $\mathcal{L}(\cdot)$ denotes the law (distribution) of a random variable. Equation
	\eqref{eq:bk_cond_dist} is the standard condition for
	independence,
	so $\mathbf{b}_k \perp \mathcal{T}_{-k}$, and consequently
	$\mathbf{b}_k \perp \mathbf{P}_{-k}^{\perp}$.
	
	Combining the three results above,
	\begin{equation}
		\boldsymbol{\zeta}_k = \mathbf{a}_k^{\perp} + \mathbf{b}_k
		\;\perp\; \mathbf{P}_{-k}^{\perp},
	\end{equation}
	so the pair $(S_k, \boldsymbol{\zeta}_k)$ is independent of $\mathbf{P}_{-k}^{\perp}$. 
	
	We now show that the cross term~\eqref{eq:cross_term} is asymptotically
	negligible relative to the leading-order signal term $\mathbb{E}[S_k^2]$ as
	$L_{pj}\to\infty$. Steps~1 and~2 below are exact, and Step~3 relies on the
	large-$L_{pj}$ approximation established in Section~\ref{sec:proof}.
	
	\medskip
	\noindent\textbf{Step 1: Factoring $\mathbf{P}_{-k}^{\perp}$ out using independence.}
	Since $(S_k, \boldsymbol{\zeta}_k) \perp \mathbf{P}_{-k}^{\perp}$, and $S_k$ is real, the tower
	property of conditional expectation gives
	\begin{align}
		\mathbb{E}\!\left[S_k\,\mathrm{Re}\!\left(
		\mathbf{e}_k^H \mathbf{P}_{-k}^{\perp} \boldsymbol{\zeta}_k\right)\right]
		&= \mathbb{E}\!\left[\mathrm{Re}\!\left(\mathbf{e}_k^H \mathbf{P}_{-k}^{\perp}\,
		\mathbb{E}\!\left[S_k\,\boldsymbol{\zeta}_k \,\big|\, \mathbf{P}_{-k}^{\perp}\right]
		\right)\right] \notag\\
		&= \mathbb{E}\!\left[\mathrm{Re}\!\left(\mathbf{e}_k^H \mathbf{P}_{-k}^{\perp}\,
		\mathbb{E}\!\left[S_k\,\boldsymbol{\zeta}_k\right]\right)\right],
		\label{eq:after_factor}
	\end{align}
	
	\medskip
	\noindent\textbf{Step 2: Computing $\mathbb{E}[S_k\,\boldsymbol{\zeta}_k]$ entry by entry.}
	Consider $\boldsymbol{\zeta}_k = \mathbf{a}_k^{\perp} + \mathbf{b}_k$, and the contribution from $\mathbf{b}_k$.
	For $n\notin\mathcal{R}_k$, $t_{nk}$ is independent of the phase shift $\phi_n$
	(designed for a different user) and $\mathbb{E}[t_{nk}]=0$, so
	$\mathbb{E}[S_k\,\mathbf{b}_k] = \mathbb{E}[S_k]\,\mathbb{E}[\mathbf{b}_k] = \mathbf{0}$.
	
	\smallskip
	Now consider $\mathbf{a}_k^{\perp}$ and the off-diagonal entries ($i\neq k$).
	Both $S_k = \sum_{n\in\mathcal{R}_k}|h_{kn}||t_{nk}|$ and
	$[\mathbf{a}_k^{\perp}]_i = \sum_{n\in\mathcal{R}_k}|h_{in}||t_{nk}|
	e^{j(\angle h_{in}-\angle h_{kn})}$ are functions of the same
	i.i.d.\ Rayleigh variables $\{|t_{nk}|\}_{n\in\mathcal{R}_k}$.
	Using $\mathbb{E}[|t_{nk}|^2]=1$ and $\mathbb{E}[|t_{nk}|]\,\mathbb{E}[|t_{mk}|]=\pi/4$ for  $n\ne m$,
	\begin{align}
		\begin{split}
			\mathbb{E}\!\left[S_k\,[\mathbf{a}_k^{\perp}]_i\right]
			&= \sum_{n,m\in\mathcal{R}_k}
			|h_{kn}|\,|h_{im}|\,e^{j(\angle h_{im}-\angle h_{km})}\,
			\mathbb{E}\!\left[|t_{nk}||t_{mk}|\right] \\
			&= \sum_{n\in\mathcal{R}_k}|h_{kn}|\,|h_{in}|\,e^{j(\angle h_{in}-\angle h_{kn})}\\
			&+ \frac{\pi}{4}\!\!\sum_{\substack{n,m\in\mathcal{R}_k\\n\neq m}}\!\!
			|h_{kn}|\,|h_{im}|\,e^{j(\angle h_{im}-\angle h_{km})} \\
			&= \left(1-\frac{\pi}{4}\right)\sum_{n\in\mathcal{R}_k}
			|h_{kn}|\,|h_{in}|\,e^{j(\angle h_{in}-\angle h_{kn})}\\
			&+ \frac{\pi}{4}
			\left(\sum_{n\in\mathcal{R}_k}|h_{kn}|\right)
			\left(\sum_{m\in\mathcal{R}_k}|h_{im}|\,e^{j(\angle h_{im}-\angle h_{km})}\right).
			\label{eq:Sk_aperp_i}
		\end{split}	
	\end{align}

	\smallskip Considering	the $k$-th entry of $\mathbf{a}_k^{\perp}$.
	By construction $[\mathbf{a}_k^{\perp}]_k=0$, so $	\mathbb{E}\!\left[S_k\,[\mathbf{a}_k^{\perp}]_k\right] = 0.$
	
	\smallskip	Collecting all contributions and using $\mathbb{E}[S_k\mathbf{b}_k]=\mathbf{0}$, $\mathbb{E}[S_k\,\boldsymbol{\zeta}_k]$ has entries
	\begin{equation}
		\left[\mathbb{E}\!\left[S_k\,\boldsymbol{\zeta}_k\right]\right]_i
		= \begin{cases}
			0 & i = k,\\[4pt]
			\mathbb{E}\!\left[S_k\,[\mathbf{a}_k^{\perp}]_i\right] & i \neq k,
		\end{cases}
		\label{eq:ESk_zeta}
	\end{equation}
	where $\mathbb{E}[S_k[\mathbf{a}_k^{\perp}]_i]$ is given by \eqref{eq:Sk_aperp_i}.
	
	\medskip
	\noindent\textbf{Step 3: Evaluating \eqref{eq:after_factor} using the rank-one projector.}
	As established in Section~\ref{sec:proof} (equation~\eqref{eq:rank1_projector}),
	$\mathbf{P}_{-k}^{\perp} = \mathbf{u}_k\mathbf{u}_k^H$.
	Substituting and denoting
	$[\mathbf{u}_k]_i = u_{ki}$, the $(k,i)$ entry of $\mathbf{u}_k\mathbf{u}_k^H$ is $u_{kk}\,u_{ki}^*$, which gives
	\begin{align}
		\mathbb{E}\!\left[\mathrm{Re}\!\left(\mathbf{e}_k^H\mathbf{P}_{-k}^{\perp}\,
		\mathbb{E}\!\left[S_k\boldsymbol{\zeta}_k\right]\right)\right]
		&= \mathrm{Re}\!\left(\mathbf{e}_k^H\,\mathbb{E}[\mathbf{u}_k\mathbf{u}_k^H]\,
		\mathbb{E}\!\left[S_k\boldsymbol{\zeta}_k\right]\right) \notag\\
		&= \mathrm{Re}\!\left(
		\sum_{i=1}^{K}\mathbb{E}[u_{kk}\,u_{ki}^*]\,
		\left[\mathbb{E}\!\left[S_k\boldsymbol{\zeta}_k\right]\right]_i
		\right).
		\label{eq:sum_form}
	\end{align}
	
	We split the sum in \eqref{eq:sum_form} into two parts.	From \eqref{eq:ESk_zeta}, the $i=k$ term vanishes exactly.  
	In the following, we show that the $i\neq k$ terms are asymptotically negligible relative to
	$\mathbb{E}[S_k^2]$.
	The remaining terms involve $\mathbb{E}[u_{kk}\,u_{ki}^*]$ and
	$[\mathbb{E}[S_k\boldsymbol{\zeta}_k]]_i$ for $i\neq k$; we bound the order of
	each factor separately.
	
	Equation~\eqref{eq:Sk_aperp_i} for $[\mathbb{E}[S_k\boldsymbol{\zeta}_k]]_i$
	($i\neq k$) has two terms, which we bound separately:
	\emph{Term 1}, $(1-\tfrac{\pi}{4})\sum_{n\in\mathcal{R}_k}|h_{kn}||h_{in}|
	e^{j(\angle h_{in}-\angle h_{kn})}$, and \emph{Term 2}, the product
	$\tfrac{\pi}{4}\bigl(\sum_{n\in\mathcal{R}_k}|h_{kn}|\bigr)
	\bigl(\sum_{m\in\mathcal{R}_k}|h_{im}|e^{j(\angle h_{im}-\angle h_{km})}\bigr)$.
	
	\smallskip
	\noindent\emph{The oscillating sums in both terms share one rate.} The sum
	in Term 1 and the second factor of Term 2 both oscillate with the phase
	$\angle h_{im}-\angle h_{km} = -\frac{2\pi}{\lambda}(r_{im}-r_{km})$, which
	is exactly~\eqref{eq:phase_diff} with $j\to k$. So the same column-major
	geometric-series analysis used to derive~\eqref{eq:offdiag_colmajor}
	applies to both, and for the confirmed antenna geometry of
	Section~\ref{sec:system} ($\alpha_{ik}^{(z)}=0\in\mathbb{Z}$ identically),
	both inherit its rate:
	\begin{equation}
		\sum_{n\in\mathcal{R}_k}|h_{kn}||h_{in}|e^{j(\angle h_{in}-\angle h_{kn})}
		= O\!\left(\frac{L_{pk}}{\sqrt{L}}\right)
		\quad\text{(Term 1's sum)},
	\end{equation}
	\begin{equation}
		\sum_{m\in\mathcal{R}_k}|h_{im}|e^{j(\angle h_{im}-\angle h_{km})}
		= O\!\left(\frac{L_{pk}}{\sqrt{L}}\right)
		\text{(Term 2's second factor)}.
	\end{equation}
	\noindent Therefore, Term 1 is $O(L_{pk}/\sqrt{L})$, and Term 2
	is a product, where its first factor $\sum_{n\in\mathcal{R}_k}|h_{kn}|=O(L_{pk})$
	is already bounded
	in~\eqref{eq:m_j_diag}, so
	$$\text{Term 2} = O(L_{pk})\cdot O\!\left(\frac{L_{pk}}{\sqrt{L}}\right)
	= O\!\left(\frac{L_{pk}^2}{\sqrt{L}}\right).$$
	Since $L_{pk}^2/\sqrt{L}$ grows faster than $L_{pk}/\sqrt{L}$ as
	$L_{pk}\to\infty$, Term 2 dominates Term 1, and their sum satisfies
	\begin{equation}
		\begin{split}
			[\mathbb{E}[S_k\boldsymbol{\zeta}_k]]_i
			&= \underbrace{O\!\left(\frac{L_{pk}}{\sqrt{L}}\right)}_{\text{Term 1}}
			+ \underbrace{\frac{\pi}{4}\,O(L_{pk})\cdot O\!\left(\frac{L_{pk}}{\sqrt{L}}\right)}_{\text{Term 2}}\\
			&= O\!\left(\frac{L_{pk}^2}{\sqrt{L}}\right),
			\qquad i\neq k.
		\end{split}
		\label{eq:ESk_zeta_order}
	\end{equation}
	For $\mathbb{E}[u_{kk}\,u_{ki}^*]$, we first quantify the qualitative
	convergence $[\mathbf{u}_k]_i\xrightarrow{p}0$ established in Step~2 of
	Section~\ref{sec:proof} (equation~\eqref{eq:uk_j_to_zero}) into an explicit
	rate, $\mathbb{E}[|u_{ki}|^2] = O(1/L_{pi})$. Squaring the deterministic
	bound~\eqref{eq:uk_j_bound} via $(a+b)^2\leq 2a^2+2b^2$ and taking
	expectations,
	\begin{equation}
		\mathbb{E}[|u_{ki}|^2]
		\leq \frac{2\,\mathbb{E}\bigl[|\mathbf{m}_i^H\mathbf{u}_k|^2\bigr]}{|[\mathbf{m}_i]_i|^2}
		+ \frac{2\cdot O(L_{pi}^2/L)}{|[\mathbf{m}_i]_i|^2}.
		\label{eq:uki_sq_bound}
	\end{equation}
	Since $\mathbf{m}_i^H\mathbf{u}_k=-\mathbf{f}_i^H\mathbf{u}_k$
	(from~\eqref{eq:null_split}), Cauchy--Schwarz gives
	$\mathbb{E}[|\mathbf{m}_i^H\mathbf{u}_k|^2]\leq\mathbb{E}[\|\mathbf{f}_i\|^2]
	=O(L)$ from~\eqref{eq:f_norm}. 	From~\eqref{eq:m_j_diag}, we can conclude the tight order $|[\mathbf{m}_i]_i|^2=\Theta(L_{pi}^2)$, because $[\mathbf{m}_i]_i$ is real and positive. Dividing each term of~\eqref{eq:uki_sq_bound} by this
	$\Theta(L_{pi}^2)$ denominator gives their individual orders:
	\begin{equation}
		\frac{O(L)}{\Theta(L_{pi}^2)} = O\!\left(\frac{L}{L_{pi}^2}\right)
		\quad\text{(Term 1)},
		\label{eq:term1_order}
	\end{equation}
	\begin{equation}
		\frac{O(L_{pi}^2/L)}{\Theta(L_{pi}^2)} = O\!\left(\frac{1}{L}\right)
		\quad\text{(Term 2, $L_{pi}^2$ cancels)}.
		\label{eq:term2_order}
	\end{equation}
	Since $L_{pi}\leq L$ always (a partition cannot exceed the array),
	$1/L\leq L/L_{pi}^2$, so Term 2's bound in~\eqref{eq:term2_order} is also
	covered by the order $O(L/L_{pi}^2)$ from Term~1. Combination of both terms
	in~\eqref{eq:uki_sq_bound} are therefore $O(L/L_{pi}^2)$, and since
	$L_{pi}=\Theta(L/K)$ ($K$ fixed), substituting $L_{pi}^2=\Theta(L^2/K^2)$
	gives
	\begin{equation}
		\mathbb{E}[|u_{ki}|^2]
		= O\!\left(\frac{L}{L_{pi}^2}\right)
		= O\!\left(\frac{L}{L^2/K^2}\right)
		= O\!\left(\frac{K^2}{L}\right)
		= O\!\left(\frac{1}{L_{pi}}\right),
		\label{eq:uki_sq_rate}
	\end{equation}
	where the last step uses $L=\Theta(K L_{pi})$ and absorbs the fixed
	constant $K$ into the $O(\cdot)$ notation.
	This rate then bounds $\mathbb{E}[u_{kk}\,u_{ki}^*]$ through the chain
	\begin{equation}
		\begin{split}
			\bigl|\mathbb{E}[u_{kk}\,u_{ki}^*]\bigr|
			\overset{(a)}{\leq} \mathbb{E}[|u_{kk}\,u_{ki}^*|]
			\overset{(b)}{\leq} \mathbb{E}[|u_{ki}|]
			\overset{(c)}{\leq}
			&
			\sqrt{\mathbb{E}[|u_{ki}|^2]}\\
			= &O\!\left(L_{pi}^{-1/2}\right),
		\end{split}
		\label{eq:cross_u_rate}
	\end{equation}
	where (a) is $|\mathbb{E}[X]|\leq\mathbb{E}[|X|]$; (b) uses
	$|u_{kk}|\leq 1$ (since $\|\mathbf{u}_k\|=1$), so
	$|u_{kk}u_{ki}^*|=|u_{kk}||u_{ki}|\leq|u_{ki}|$; and (c) is
	Cauchy--Schwarz applied to $\mathbb{E}[|u_{ki}|\cdot 1]$, using the rate
	from~\eqref{eq:uki_sq_rate}.
	Since all partition sizes grow together, $L_{pi}=\Theta(L_{pk})=\Theta(L/K)$
	(equal or comparable partitioning), so \eqref{eq:cross_u_rate} is
	$O(L_{pk}^{-1/2})$. Combining \eqref{eq:ESk_zeta_order} and \eqref{eq:cross_u_rate}, each term in
	the $i\neq k$ sum of \eqref{eq:sum_form} satisfies
	\begin{equation}
		\bigl|\mathbb{E}[u_{kk}\,u_{ki}^*]\,[\mathbb{E}[S_k\boldsymbol{\zeta}_k]]_i\bigr|
		= O(L_{pk}^{-1/2})\cdot O\!\left(\frac{L_{pk}^2}{\sqrt{L}}\right)
		= O\!\left(\frac{L_{pk}^{3/2}}{\sqrt{L}}\right).
		\label{eq:cross_term_order}
	\end{equation}
	
	\bibliographystyle{IEEEtran}
	\bibliography{references}

\begin{thebibliography}{10}
\providecommand{\url}[1]{#1}
\csname url@samestyle\endcsname
\providecommand{\newblock}{\relax}
\providecommand{\bibinfo}[2]{#2}
\providecommand{\BIBentrySTDinterwordspacing}{\spaceskip=0pt\relax}
\providecommand{\BIBentryALTinterwordstretchfactor}{4}
\providecommand{\BIBentryALTinterwordspacing}{\spaceskip=\fontdimen2\font plus
\BIBentryALTinterwordstretchfactor\fontdimen3\font minus
  \fontdimen4\font\relax}
\providecommand{\BIBforeignlanguage}[2]{{%
\expandafter\ifx\csname l@#1\endcsname\relax
\typeout{** WARNING: IEEEtran.bst: No hyphenation pattern has been}%
\typeout{** loaded for the language `#1'. Using the pattern for}%
\typeout{** the default language instead.}%
\else
\language=\csname l@#1\endcsname
\fi
#2}}
\providecommand{\BIBdecl}{\relax}
\BIBdecl

\bibitem{DiRenzo2020}
M.~Di~Renzo \emph{et~al.}, ``Smart radio environments empowered by
  reconfigurable intelligent surfaces: How it works, state of research, and the
  road ahead,'' \emph{IEEE J. Sel. Areas Commun.}, vol.~38, no.~11, pp.
  2450--2525, Nov. 2020.

\bibitem{Zhao2025}
X.~Zhao \emph{et~al.}, ``Reconfigurable intelligent surfaces for 6{G}:
  Engineering challenges and the road ahead,'' \emph{Intell. Converg. Netw.},
  vol.~6, no.~1, pp. 53--81, Mar. 2025.

\bibitem{Huang2020Holographic}
C.~Huang \emph{et~al.}, ``Holographic {MIMO} surfaces for 6{G} wireless
  networks: Opportunities, challenges, and trends,'' \emph{IEEE Wireless
  Commun.}, vol.~27, no.~5, pp. 118--125, Oct. 2020.

\bibitem{Li2025Transmissive}
Z.~Li \emph{et~al.}, ``Transmissive {RIS} transceiver enabled multistream
  communication systems: Design, optimization, and analysis,'' \emph{IEEE
  Internet Things J.}, vol.~12, no.~5, pp. 5985--6000, Mar. 2025.

\bibitem{Li2025TMA}
------, ``Toward {TMA}-based transmissive {RIS} transceiver enabled downlink
  communication networks: A consensus-{ADMM} approach,'' \emph{IEEE Trans.
  Commun.}, vol.~73, no.~4, pp. 2832--2846, Apr. 2025.

\bibitem{Mao2024Index}
T.~Mao \emph{et~al.}, ``Index-modulation-aided terahertz communications with
  reconfigurable intelligent surface,'' \emph{IEEE Trans. Wireless Commun.},
  vol.~23, no.~7, pp. 8059--8070, Jul. 2024.

\bibitem{Huang2024Hybrid}
A.~Huang \emph{et~al.}, ``Hybrid active-passive {RIS} transmitter enabled
  energy-efficient multi-user communications,'' \emph{IEEE Trans. Wireless
  Commun.}, vol.~23, no.~9, pp. 10\,653--10\,666, Sep. 2024.

\bibitem{Zeng2024Reconfigurable}
S.~Zeng \emph{et~al.}, ``Reconfigurable refractive surface-enabled multi-user
  holographic {MIMO} communications,'' \emph{IEEE Trans. Wireless Commun.},
  vol.~23, no.~5, pp. 4845--4860, May 2024.

\bibitem{SunJing2024Holographic}
Z.~Sun and Y.~Jing, ``Holographic {MIMO} {NOMA} communications: A power saving
  design,'' \emph{IEEE Trans. Wireless Commun.}, vol.~23, no.~12, pp.
  18\,711--18\,724, Dec. 2024.

\bibitem{Jung2019Reliability}
M.~Jung \emph{et~al.}, ``Reliability analysis of large intelligent surfaces
  ({LIS}s): Rate distribution and outage probability,'' \emph{IEEE Wireless
  Commun. Lett.}, vol.~8, no.~6, pp. 1662--1666, Dec. 2019.

\bibitem{Yuan2020Towards}
J.~Yuan \emph{et~al.}, ``Towards large intelligent surface ({LIS})-based
  communications,'' \emph{IEEE Trans. Commun.}, vol.~68, no.~10, pp.
  6568--6582, Oct. 2020.

\bibitem{Jung2020Performance}
M.~Jung \emph{et~al.}, ``Performance analysis of large intelligent surfaces
  ({LIS}s): Asymptotic data rate and channel hardening effects,'' \emph{IEEE
  Trans. Wireless Commun.}, vol.~19, no.~3, pp. 2052--2065, Mar. 2020.

\bibitem{Jung2021PerformanceActive}
------, ``Performance analysis of active large intelligent surfaces ({LIS}s):
  Uplink spectral efficiency and pilot training,'' \emph{IEEE Trans. Commun.},
  vol.~69, no.~5, pp. 3379--3394, May 2021.

\bibitem{Hu2018Beyond}
S.~Hu \emph{et~al.}, ``Beyond massive {MIMO}: The potential of data
  transmission with large intelligent surfaces,'' \emph{IEEE Trans. Signal
  Process.}, vol.~66, no.~10, pp. 2746--2758, May 2018.

\bibitem{Sanchez2022Distributed}
J.~R. S{\'a}nchez \emph{et~al.}, ``Distributed and scalable uplink processing
  for {LIS}: Algorithm, architecture, and design trade-offs,'' \emph{IEEE
  Trans. Signal Process.}, vol.~70, pp. 2639--2653, 2022.

\bibitem{Li2024TowardTransmissive}
Z.~Li \emph{et~al.}, ``Toward transmissive {RIS} transceiver enabled uplink
  communication systems: Design and optimization,'' \emph{IEEE Internet Things
  J.}, vol.~11, no.~4, pp. 6788--6801, Feb. 2024.

\bibitem{An2024Exploiting}
K.~An \emph{et~al.}, ``Exploiting multi-layer refracting {RIS}-assisted
  receiver for {HAP}-{SWIPT} networks,'' \emph{IEEE Trans. Wireless Commun.},
  vol.~23, no.~10, pp. 12\,638--12\,657, Oct. 2024.

\bibitem{Sun2025SINR}
Z.~Sun \emph{et~al.}, ``{SINR} analysis and interference mitigation for
  continuous-aperture holographic {MIMO} uplink,'' \emph{IEEE Trans. Veh.
  Technol.}, vol.~74, no.~1, pp. 1747--1752, Jan. 2025.

\bibitem{RahimianJingYu2026}
S.~Rahimian, Y.~Jing, and X.~Yu, ``{RIS}-equipped receiver for two-user systems
  using surface partitioning and {ZF} beamforming,'' \emph{IEEE Wireless
  Commun. Lett.}, vol.~15, pp. 1335--1339, 2026.

\bibitem{Chauhan2023RIS}
A.~Chauhan \emph{et~al.}, ``{RIS} partition-assisted non-orthogonal multiple
  access ({NOMA}) and quadrature-{NOMA} with imperfect {SIC},'' \emph{IEEE
  Trans. Wireless Commun.}, vol.~22, no.~7, pp. 4371--4386, Jul. 2023.

\bibitem{KhaleelBasar2022NOMA}
A.~Khaleel and E.~Basar, ``A novel {NOMA} solution with {RIS} partitioning,''
  \emph{IEEE J. Sel. Topics Signal Process.}, vol.~16, no.~1, pp. 70--81, Jan.
  2022.

\bibitem{vanTrees2002}
H.~L. van Trees, \emph{Optimum Array Processing}.\hskip 1em plus 0.5em minus
  0.4em\relax Wiley-Interscience, 2002.

\bibitem{Kay1993}
S.~M. Kay, \emph{Fundamentals of Statistical Signal Processing, Vol. I}.\hskip
  1em plus 0.5em minus 0.4em\relax Prentice Hall, 1993.

\bibitem{Billingsley1995}
P.~Billingsley, \emph{Probability and Measure}, 3rd~ed.\hskip 1em plus 0.5em
  minus 0.4em\relax Wiley, 1995.

\bibitem{tse2005fundamentals}
D.~Tse and P.~Viswanath, \emph{Fundamentals of Wireless Communication}.\hskip
  1em plus 0.5em minus 0.4em\relax Cambridge Univ. Press, 2005.

\end{thebibliography}
	
\end{document}